\documentclass[12pt]{article}

\advance \topmargin by -\headheight
\advance \topmargin by -\headsep

\evensidemargin \oddsidemargin
\usepackage{ulem}
\usepackage[OT1]{fontenc}
\usepackage{amsthm,amsmath}
\usepackage[colorlinks,linkcolor = blue, citecolor=blue,urlcolor=blue]{hyperref}

\usepackage{amssymb}
\usepackage{tikz}
\usetikzlibrary{arrows}
\usetikzlibrary{decorations.pathreplacing}

\tikzstyle{block} = [rectangle, draw, fill=black!10,text width=10em, text centered, rounded corners, minimum height=1em]
\tikzstyle{bb} = [rectangle, draw, fill=black!10, text width=15em, text centered, rounded corners, minimum height=1em]

\usepackage{mathrsfs,amsfonts}
\usepackage{mathtools}
\usepackage{natbib}
\usepackage{booktabs}
\usepackage{multirow}
\usepackage[center,footnotesize]{caption}
\usepackage{graphicx}
\usepackage{caption}
\usepackage{subcaption}
\usepackage{diagbox}
\usepackage{xcolor}
\usepackage[T1]{fontenc}
\usepackage[vlined,ruled]{algorithm2e}
\usepackage{cases}
\usepackage{xfrac}
\usepackage[capitalise]{cleveref}
\usepackage{enumitem}

\usepackage{siunitx}

\usepackage{textcomp}

\usepackage{mathptmx,times}
\usepackage[final]{microtype}

\usepackage[mathlines,pagewise,modulo]{lineno}

\numberwithin{equation}{section}
\theoremstyle{plain}

\newtheorem{theorem}{Theorem}
\newtheorem{proposition}{Proposition}
\newtheorem{lemma}{Lemma}
\newtheorem{corollary}{Corollary}

\newtheorem{remark}{Remark}

\DeclareMathOperator*{\argmin}{arg\,min}

\newcommand{\bA}{\mathbf{A}}

\newcommand{\bH}{\mathbf{H}}
\newcommand{\bI}{\mathbf{I}}
\newcommand{\bJ}{\mathbf{J}}

\newcommand{\bP}{\mathbf{P}}
\newcommand{\bQ}{\mathbf{Q}}

\newcommand{\bU}{\mathbf{U}}
\newcommand{\bV}{\mathbf{V}}

\newcommand{\bZ}{\mathbf{Z}}

\newcommand{\ba}{\mathbf{a}}
\newcommand{\bb}{\mathbf{b}}

\newcommand{\bg}{\mathbf{g}}

\newcommand{\bw}{\mathbf{w}}
\newcommand{\bx}{\mathbf{x}}
\newcommand{\by}{\mathbf{y}}
\newcommand{\bz}{\mathbf{z}}
\newcommand{\bbx}{\bar{\bx}}

\newcommand{\ev}{\mathbb{E}}
\newcommand{\var}{\mathrm{var}}

\newcommand{\real}{\mathbb{R}}

\newcommand{\rn}[1]{%
  \textup{\uppercase\expandafter{\romannumeral#1}}%
}

\newcommand{\T}{\scriptscriptstyle{\text{T}}}

\newcommand{\vech}{\text{vech}}
\newcommand{\ve}{\text{vec}}
\newcommand{\spn}{\text{span}}
\usepackage{graphicx}

\newcommand{\n}{\nonumber}

\pdfoutput=1

\begin{document}

\title{\bf Enveloped Huber Regression}

\author{\sc{Le Zhou, R. Dennis Cook and Hui Zou} \\
School of Statistics \\
      University of Minnesota
 }
     
\date{\today}
\maketitle

\begin{abstract}
Huber regression (HR) is a popular robust alternative to the least squares regression when the error follows a heavy-tailed distribution. We propose a new method called the enveloped Huber regression (EHR) by considering the envelope assumption that there exists some subspace of the predictors that has no association with the response, which is referred to as the immaterial part. More efficient estimation is achieved via the removal of the immaterial part. Different from the envelope least squares (ENV) model whose estimation is based on maximum normal likelihood, the estimation of the EHR model is through Generalized Method of Moments. The asymptotic normality of the EHR estimator is established, and it is shown that EHR is more efficient than HR. Moreover, EHR is more efficient than ENV when the error distribution is heavy-tailed, while maintaining a small efficiency loss when the error distribution is normal. Moreover, our theory also covers the heteroscedastic case in which the error may depend on the covariates. Extensive simulation studies confirm the messages from the asymptotic theory. EHR is further illustrated on a real dataset.

\vskip10pt

  \noindent{\bf Keywords:} {Asymptotics efficiency; Envelope model; Huber regression; Heavy-tailed distributions}

\end{abstract}

\section{Introduction}
Envelope models  have drawn considerable attentions over the past ten years. The idea was first proposed in \cite{cook2007} to deal with certain high dimensional problems in sufficient dimension reduction.  They \citep{cook2010} later extended envelopes to dimension reduction in response variables in multivariate linear regression model with normal errors. Under this setting, the envelope model has been shown to effectively reduce estimation variance compared with standard methods. The envelope approach was extended to a variety of envelope models \citep{su2011,su2012,su2013}. The predictor envelope was proposed and studied in \cite{cook2013} in which a connection between the envelope model and partial least squares was revealed. This led to further studies on partial least squares in the high-dimensional setting \citep{cook2019,su2020}.  \cite{cook2015} applied the envelope method to reduced rank regression, and \cite{zhang2015} extended the envelope method from linear regression to generalized linear models and Cox's proportional hazard model. \cite{su2016} proposed sparse envelope models for variable selection in the envelope model for multivariate regression. 
\cite{ding2017} developed an enveloped quantile regression model \citep{kb1978, Koenker2005}. A Bayesian approach to the envelope model was studied in \cite{su2017}. \cite{zhang2017} and \cite{dingcook2018} developed  envelope models for tensor and matrix regression problems. Spatial envelopes and envelopes for time series data were proposed and studied in \cite{R2017} and \cite{wang2018}. For a comprehensive overview of envelope methods, see \cite{cook2018}.

In this paper we revisit the classical setting for using the envelope model, i.e., the multiple linear regression problem. The existing envelope methods are developed by using the normal likelihood and their optimal efficiency gain over the standard regression model relies on the normal error assumption. Hence, it is not clear whether the likelihood-based envelope model is the best choice when the error follows some heavy tailed distribution. This motivated us to consider a non-likelihood based envelope model that is also resistant to heavy tailed errors. Huber regression \citep{Huber1964} is the celebrated robust alternative to the least squares because it enjoys high efficiency and robustness simultaneously.   \cite{Huber1964} accomplished this by replacing the squared error loss in least squares with $\rho(r)$, where $r$ denotes the residual and the derivative of $\rho(r)$ is $\psi(r) =\max\{\min\{r,k\},-k\}$ where $k$ is some constant. We aim to extend the envelope idea to the Huber regression model and we call the resulting method  enveloped Huber regression (or EHR for short). There are dual purposes of doing so. First, we want to improve the robustness of the classical envelope regression model. Second, we also want to reduce the estimation variance of the classical Huber regression model. The main idea of EHR comes from examining $\ev^{\rho}[y|\bx] = \argmin_{q\in\real}\ev[\rho(y-q)|\bx]$, the target in Huber regression, from the perspective of envelope modeling. It is quite reasonable to assume that $\ev^{\rho}[y|\bx]$ is actually $\ev^{\rho}[y|\bP_S\bx]$ where $\bP_S$ is the projection matrix to the subspace $S$ in the predictor space. 
This assumption means that there are linear combinations of the predictors that are irrelevant in the model, which is referred to as the immaterial part. 
If we can effectively detect and remove the immaterial part, a significant gain in estimation efficiency is achieved. The classical envelope model was formulated and estimated by using the normal-likelihood under the assumption that the error is Gaussian and independent of $\bx$. We would like to avoid using any parametric distribution to model the error term.  
We estimate EHR through generalized method of moments (GMM) \citep{lars1982}, in which no specific error distribution is assumed. We establish asymptotic normality of the proposed EHR estimator and demonstrate its efficiency gains compared with the classical Huber regression and the classical envelope model. Moreover, our theory also covers the conditional heteroscedastic case in which the error may depend on the covariates, while the existing work are based on the independent homoscedastic error assumption \citep{cook2013,cook2019}.
Numerical studies and a real data example confirm the efficiency gains in finite samples. 

The paper is organized as follows. In Section 2, we review the basics of the Huber regression and then derive the formulation of the enveloped Huber regression (EHR). In Section 3, we discuss the estimation procedure of EHR via GMM. In Section 4, we establish the asymptotic distribution of EHR estimators and we show that it is more efficient than Huber regression. Simulation studies are presented in Section 5, and a real data example is given in Section 6. The proofs of all the theoretical results are shown in an Appendix.

\section{The Model}
\subsection{Notation}
To facilitate the discussion, we first introduce some notation and definitions. 
Let $\by$ and $\bx$ be the response and covariates, respectively.
Let $\real^{p\times u}$ be the set of all $p\times u$ matrices. Let $\bI_p$ be the $p\times p$ identity matrix. For a subspace $S$ in $\real^p$, let $S^{\perp}$ be its orthogonal complement in $\real^p$. For any matrix $M \in \real^{p\times u}$, $\spn(M)$ stands for the subspace spanned by the columns of $M$. Let $\bP_M$ be
the projection matrix onto $\spn(M)$ and $\bQ_M = \bI_p - \bP_M$ be the projection matrix onto its orthogonal complement $\spn(M)^{\perp}$. Similarly, for a subspace $S$ in $\real^p$,  let $\bP_S$ and $\bQ_S$ be the projection matrix onto $S$ and $S^{\perp}$, respectively.  ``vec" stands for the vectorization of a matrix column-wise, while ``vech'' stands for the half-vectorization operator that vectorizes the lower triangle of a symmetric matrix. Therefore, $M\in \real^{p\times p}$, $\ve(M)\in\real^{p^2}$ and $\vech(M)\in\real^{p(p+1)/2}$. Let $\|M\|$ be the operator norm of a matrix $M$, and let $M^{\dagger}$ represent the Moore-Penrose inverse of a positive semidefinite matrix $M$. Let $\overset{d}\to$ and $\overset{P}\to$ stand for convergence in distribution and probability, respectively. We say a subspace $R$ in $\real^p$ is a reducing subspace of a matrix $M\in \real^{p\times p}$ whenever $M = \bP_R M \bP_R + \bQ_R M \bQ_R$. For two symmetric matrices $A$ and $B$,  $A\preceq B$ means that $B-A$ is positive semidefinite and $A\prec B$ means that $B-A$ is positive definite. Finally, for two matrices $A\in\real^{m\times n}$ and $B\in\real^{p\times q}$, let $A\otimes B$ be their Kronecker product which is a $pm\times qn$ matrix.

\subsection{Huber regression}\label{huber}
Huber regression was first proposed in \cite{Huber1964}. It serves as an alternative tool to the traditional ordinary least squares, which can be sensitive to outliers and inefficient under heavy-tailed error distributions. The Huber loss is defined as $\rho(x) = \begin{cases}
k|x| - \frac{k^2}{2} &\text{$|x| > k$}\\
\frac{x^2}{2} &\text{$|x|\le k $}
\end{cases}$, with the corresponding influence function being $\psi(x)= \dot{\rho}(x) = \begin{cases}
k  &\text{$x > k$}\\
x &\text{$|x|\le k $}\\
-k &\text{$x\le -k $}\\
\end{cases}$. Here $k$ is a tuning parameter, which will be discussed later. Huber regression aims to estimate the following quantity, $\ev^{\rho}[y|\bx] = \argmin_{u\in\real}\ev[\rho(y-u) |\bx]$, which is referred to as ``Huber's $\rho$-mean''. A linear Huber regression model assumes that $\ev^{\rho}[y|\bx] = \mu^*+\bx^{\T}\beta^*$ for some unknown $\mu^*\in\real$ and $\beta^*\in\real^p$. By definition and boundedness of $\psi(\cdot)$, the linear Huber regression model can be written as follows 
\begin{equation}\label{Hubermodel}
y = \mu^*+\bx^{\T}\beta^* + \epsilon \text{ with } \ev[\psi(\epsilon)|\bx]=0.
\end{equation}
We use the notation $\bx \sim (\mu_x^*,\Sigma_x^*)$ to denote that $\bx$ follows a distribution with mean $\mu_x^*$ and some positive definite covariance matrix $\Sigma_x^*$.

Compare model~(\ref{Hubermodel}) with the standard linear regression model $y = \mu^*+\bx^{\T}\beta^* + \epsilon$ where $\epsilon$ is independent of $\bx$ and $\ev[\epsilon]$=0. It is easy to see that when $\epsilon$ in the linear regression model follows a symmetric distribution then the Huber regression model holds as well. Throughout this paper our theoretical analysis is based on model~(\ref{Hubermodel}) in which conditional heteroscedasticity is allowed (i.e., $\epsilon$ depends on $\bx$).

Given a random sample of $n$ i.i.d. observations $\{(y_i,\bx_i)\}_{i=1}^n$ from model~(\ref{Hubermodel}), 
Huber regression estimates $\mu^*$ and $\beta^*$ through 
\begin{equation*}
\min_{\mu \in\real, \beta\in\real^p } \sum_{i=1}^{n}\rho(y_i -\mu -  \bx_i^{\T}\beta).
\end{equation*}
In Huber regression, the tuning parameter $k$ is related to the efficiency. When $k\to \infty$, Huber regression is the same as OLS; when $k\to 0$, Huber regression is the same as $L_1$ regression (or median regression). It is shown by Huber that when error follows normal distribution with $\sigma=1$, $k=1.345$ produces $95\%$ efficiency relative to OLS \citep{Huber2004}. Meanwhile, it gives substantial resistance to outliers in non-normal case. A rule of thumb is to choose $k=1.345\frac{\text{MAD}}{0.6745}$, where $\text{MAD}= \text{median}|y_i - \hat\mu - \bx_i^{\T}\hat\beta|$ and $\hat\mu$ and $\hat\beta$ are estimators by median regression. We used this rule in our numerical experiments. 

\subsection{Enveloped Huber Regression}
The idea of EHR comes from the intuition that certain linear combinations of the predictors in data may be irrelevant to the conditional $\rho$-mean $\ev^{\rho}[y|\bx]$. We aim to find a subspace $S$ such that ($\ba$): $\bQ_S\bx$ and $\bP_S\bx$ are uncorrelated; and ($\bb$): the conditional Huber's $\rho$-mean of $y$ given $\bx$ depends on $\bx$ only through $\bP_S \bx$. As such, we call $\bP_S\bx$ the material part of $\bx$ and call $\bQ_S\bx$ the immaterial part of $\bx$. Condition ($\ba$) can be mathematically formulated as $\text{cov}(\bQ_S\bx, \bP_S\bx)=0$. Condition ($\bb$) links the conditional $\rho$-mean of response given covariates to the material part. It can be conceptually formulated by two possible definitions: (i), the conditional $\rho$-mean of $y$ given $\bx$ is identical with the conditional $\rho$-mean of $y$ given $\bP_S\bx$; (ii), the signal vector $\beta^*$ lies in the subspace $S$. It turns out that these two definitions are equivalent under model \eqref{Hubermodel}, which is summarized in the following proposition.
\begin{proposition}\label{equivdef}
Under model \eqref{Hubermodel}, the following two statements (\rm{i}) and (\rm{ii}) are equivalent:
$$(\rm{i}) \  \ev^{\rho}[y| \bx] = \ev^{\rho}[y| \bP_S\bx]; \quad (\rm{ii}) \  \beta^*\in S.$$
\end{proposition}
The proof for this proposition is in fact not trivial and is placed in Appendix. The proof also shows that the equivalence between these two formulations of condition ($\bb$) is actually independent of condition ($\ba$).

Therefore, condition ($\bb$) is equivalent to $\beta^*\in S$. Because $\bx \sim (\mu_x^*,\Sigma_x^*)$, the condition ($\ba$) holds if and only if $S$ is a reducing subspace of $\Sigma_x^*$ \citep{cook2010}. If we can find such a subspace $S$, then we can remove the immaterial part $\bQ_S \bx$ and hence improve the estimation efficiency in subsequent analysis. Notice that there might be two different subspaces that both satisfy conditions ($\ba$) and ($\bb$). By similar arguments in \cite{cook2010}, we take the intersection of all subspaces satisfying conditions ($\ba$) and ($\bb$) and the resulting subspace is the smallest subspace that satisfy these conditions. The smallest subspace is referred to as the $\Sigma_x^*$-predictor envelope of $\text{span}(\beta^*)$, denoted as $\mathcal{E}_{\Sigma_{x}^*}\left(\beta^*\right)$. The dimension of $\mathcal{E}_{\Sigma_{x}^*}\left(\beta^*\right)$ is denoted as $u$ where $1\le u \le p$. If $u$ is $p$, then $\real^p$ is the only subspace that can satisfy conditions ($\ba$) and ($\bb$), and the method reduces to the usual Huber regression. When $u$ is much smaller than $p$, which is very likely in many applications, we expect to have great efficiency gain via the removal of immaterial part from the model. 

\cite{cook2013} based their starting point for predictor envelopes on the pair of conditions
\begin{align}\label{condc}
\text {(i)}\; \text{cov}(\bQ_S\bx, \bP_S\bx)=0 \;\;\;\text {and\;\;\; (ii)}\; \text{cov}(y, \bQ_S\bx \mid \bP_S\bx)=0.
\end{align}
Condition (i) here is the same as our condition ($\ba$). Under the assumption that error is independent from the covariates, \cite{cook2013} derived that condition (ii) in (\ref{condc}) is equivalent to $\beta^{*}\in S$. In contrast, our approach that directly links the conditional $\rho$-mean to the material part works under more general settings, where errors do not need to be independent from the covariates. Our development diverges from this point.

Let $\Gamma\in \real^{p\times u}$ be some semi-orthogonal basis of $\mathcal{E}_{\Sigma_{x}^*}\left(\beta^*\right)$, and let $\Gamma_0\in \real^{p\times (p-u)}$ be some semi-orthogonal basis of $\mathcal{E}_{\Sigma_{x}^*}\left(\beta^*\right)^{\perp}$. Based on the conditions ($\ba$) and ($\bb$), we can do a reparameterization as follows
\begin{align}\label{env}
&\beta^* = \Gamma\eta \text{ for some $\eta \in\real^u$}; \nonumber\\
&\Sigma_x^* = \Gamma \Omega \Gamma^{\T} + \Gamma_0 \Omega_0 \Gamma_0^{\T} \text{ for some positive semidefinite matrices $\Omega$ and $\Omega_0$}. 
\end{align}
In view of \eqref{env}, we can reparameterize the original Huber regression model~(\ref{Hubermodel}) into the following enveloped Huber regression model:
\begin{align}\label{setup}
&y_i =\mu^*+\bx_i^{\T}\beta^* + \epsilon_i,\; \{(\bx_i,\epsilon_i)\}_{i=1}^{n}, (\bx,\epsilon) \text{ are i.i.d.} ,\bx \sim (\mu_x^*,  \Sigma_x^*), \ev_{\theta^*}[\psi(\epsilon)|\bx]=0\nonumber\\
&\beta^*= \Gamma\eta \nonumber\\
&\Sigma_x^* = \Gamma \Omega \Gamma^{\T} + \Gamma_0 \Omega_0 \Gamma_0^{\T} .
\end{align}
When $\Gamma$ is known, estimation of $\beta^{*}=\Gamma\eta$ is through minimizing
\begin{align}\label{eta}
\min_{\mu\in\real,\eta\in\real^u}\sum_{i=1}^n \rho(y_i - \mu-\eta^{\T}\Gamma^{\T}\bx_i),
\end{align}
which is the usual Huber regression. 
When $\Gamma$ is known, the envelope estimator $\hat\beta_{eh}^* \coloneqq  \Gamma\hat\eta$ is more efficient than the usual Huber estimator $\hat\beta_{h}$. The following proposition uses the independent error case as an example to demonstrate this point. The efficiency gain claim is also generally true for the heteroscedastic error case. See section 4 for details.

\begin{proposition}\label{p1}
Assume \eqref{setup} holds and $\hat\eta$ is the solution from \eqref{eta}. Further assumes $\epsilon$ is independent from $\bx$. 
Then $\sqrt{n}(\hat\beta_{eh}^*-\beta^*)$ converges to a multivariate normal distribution with mean zero and covariance matrix 
\begin{equation*}
\mathrm{avar}(\sqrt{n}\hat\beta_{eh}^*) = \frac{\ev_{\theta^*}[\psi^2(\epsilon)]}{(\ev_{\theta^*}[\psi'(\epsilon)])^2} \Gamma\Omega^{-1}\Gamma^{\T}.
\end{equation*}
\end{proposition}

\begin{remark}\label{rmk1}
By the standard Huber regression asymptotic theory we have
\begin{equation*}
\sqrt{n} (\hat\beta_h - \beta^*) \overset{d}\to N\left(0,\frac{\ev_{\theta^*}[\psi^2(\epsilon)]}{(\ev_{\theta^*}[\psi'(\epsilon)])^2}\Sigma_x^{*-1}\right).
\end{equation*}
By the relation in \eqref{setup}, we have $\Sigma_x^{*-1} = \Gamma \Omega^{-1} \Gamma^{\T} + \Gamma_0 \Omega_0^{-1} \Gamma_0^{\T}$. Combining this with Proposition~\ref{p1}, we see that the enveloped Huber estimator is more efficient than the Huber estimator.  This also shows that our approach is consistent with the general theory of envelope construction by \cite{zhang2015}.
\end{remark}

\section{Estimation}
We have introduced the formulation of enveloped Huber regression and demonstrated its asymptotic efficiency over Huber regression when $\Gamma$ is known. In practice, we do not know $\Gamma$, so we must include it in the estimation procedure. Note that we do not assume the error distribution is known.

We denote the unknown parameter in classical Huber regression as $\theta^*=(\mu^*,\beta^{*\T},\vech(\Sigma_x^*)^{\T},\mu_x^{*\T})^{\T}$.
Note that the minimizer of the standard Huber regression, denoted as $(\hat\mu_h,\hat\beta_h^{\T})^{\T}$, satisfies 
\begin{equation*}
\sum_{i=1}^{n}\psi(y_i - \hat\mu_h - \bx_i^{\T}\hat{\beta}_h)(1,\bx_i^{\T})^{\T} = 0.
\end{equation*}
Since $\{\bx_i\}_{i=1}^n$ are i.i.d., natural estimates of $\mu_x^*$ and $\Sigma_x^*$ are $\bbx = \frac{1}{n}\sum_{i=1}^n \bx_i$ and $S_x = \frac{1}{n}\sum_{i=1}^{n}(\bx_i - \bbx)(\bx_i-\bbx)^{\T}$. Write $\tilde\theta = (\hat\mu_h,\hat\beta_h^{\T},\vech(S_x)^{\T},\bbx^{\T})^{\T}$. We can view $\tilde\theta = (\hat\mu_h,\hat\beta_h^{\T},\vech(S_x)^{\T},\bbx^{\T})^{\T}$ as a solution to the following estimation equations
\begin{align}\label{hn}
G_n(\theta)\triangleq\begin{pmatrix}
   \frac{1}{n}\sum_{i=1}^{n} \psi(y_i - \mu - \beta^{\T}\bx_i)(1,\bx_i^{\T})^{\T}\\
  \vech(\Sigma_x) - \vech(\frac{1}{n}\sum_{i=1}^n(\bx_i - \mu_x)(\bx_i-\mu_x)^{\T})\\
  \mu_x - \bbx
 \end{pmatrix} = 0.
\end{align}
We can rewrite the above as $\frac{1}{n}\sum_{i=1}^{n}\bg(\bz_i;\theta)=0$, where we denote 
\begin{align}\label{gs}
\bz_i &= (y_i,\bx_i^{\T})^{\T}, \n\\
\bg(\bz_i;\theta) &= (\bg_1(\bz_i;\theta)^{\T},\bg_2(\bz_i;\theta)^{\T},\bg_3(\bz_i;\theta)^{\T})^{\T}, \n\\
\bg_1(\bz_i;\theta) &= \psi(y_i - \mu - \bx_i^{\T}\beta)(1,\bx_i^{\T})^{\T},\n\\
\bg_2(\bz_i;\theta) &= \vech(\Sigma_x) - \vech((\bx_i-\mu_x)(\bx_i - \mu_x)^{\T}),\n\\
\bg_3(\bz_i;\theta) &= \mu_x - \bx_i.
\end{align} 

For the envelope model \eqref{setup}, we similarly consider the corresponding estimating equations \eqref{gee}:
\begin{align}\label{gee}
G_n(\theta)=\begin{pmatrix}
   \frac{1}{n}\sum_{i=1}^{n} \psi(y_i - \mu - \eta^{\T}\Gamma^{\T}\bx_i)(1,\bx_i^{\T})^{\T}\\
  \vech(\Gamma\Omega\Gamma^{\T}+\Gamma_0\Omega_0\Gamma_0^{\T}) - \vech(\frac{1}{n}\sum_{i=1}^n(\bx_i - \mu_x)(\bx_i-\mu_x)^{\T})\\
  \mu_x - \bbx
 \end{pmatrix}= 0. 
\end{align}
Based on the above equations, it is more convenient to reparameterize the model by using $\zeta \coloneqq (\mu,\eta^{\T},\ve(\Gamma)^{\T},\vech(\Omega)^{\T},\vech(\Omega_0)^{\T},\mu_x^{\T})^{\T}$. Note that 
$\theta = (\mu,\beta^{\T},\vech(\Sigma_x)^{\T},\mu_x^{\T})^{\T}=env(\zeta)$, where $env(\zeta) \coloneqq (\mu, (\Gamma\eta)^{\T},  \vech(\Gamma\Omega\Gamma^{\T}+\Gamma_0\Omega_0\Gamma_0^{\T})^{\T},\mu_x^{\T})$ represents the relation between $\theta$ and $\zeta$ under the envelope model \eqref{setup}.

Notice that $\zeta$ has dimension $1+p+u+\frac{p(p+1)}{2}$, but number of equations we have in \eqref{gee} is $1+2p+\frac{p(p+1)}{2}$. So \eqref{gee} may not have a solution. This possibility can be seen also by inspecting the three blocks of equations in \eqref{gee}.  The last block has a unique solution $ \mu_x = \bbx$.  Substituting this into the second block,  $\Gamma$ will be a solution if and only if its columns form a basis for a reducing subspace of $S_{x}$. However, the first block of equations may not have a solution with $\Gamma$ restricted in this way.

Instead of basing estimation on \eqref{gee}, we consider the following enveloped Huber estimator through Generalized Method of Moments (GMM):
\begin{align}\label{op}
\hat\theta_{g} = \argmin_{\theta = env(\zeta)} G_n^{\T}(\theta)\hat\Delta G_n(\theta)
\end{align}
where $\hat\Delta$ is chosen to be any consistent estimator of $(\ev_{\theta^*}[\bg(\bz;\theta^*)\bg(\bz;\theta^*)^{\T}])^{-1}$, $\theta^*$ is the true value for $\theta$, and $\bz = (y,\bx^{\T})^{\T}$. For example, $\hat\Delta = (\frac{1}{n} \sum_{i=1}^{n} \bg(\bz_i;\tilde\theta)\bg(\bz_i;\tilde\theta)^{\T})^{-1}$ where $\tilde\theta$ is the previously defined estimator from GEE. An interpretation of \eqref{op} is that we minimize the objective function $G_n^{\T}(\theta)\hat\Delta G_n(\theta)$ under the envelope model constraint, namely, $\theta$ is constrained to be equal to $env(\zeta)$ for some $\zeta$. If the constraint was removed, we would get the classical GEE estimator $\tilde\theta$ which makes the objective value exactly zero.

We can see from \eqref{setup} that $\Gamma$ is not estimable, since it can be any orthogonal basis of $\mathcal{E}_{\Sigma_{x}^*}\left(\beta^*\right)$. In fact, the estimable target is $\mathcal{E}_{\Sigma_{x}^*}\left(\beta^*\right) = \text{span}(\Gamma)$. Since $\Gamma$ is constrained to be a semi-orthogonal matrix, it seems that we need to do optimization on a Grassmann manifold, which is really computation intensive.  \cite{cook2016} proposed a reparameterization of $\Gamma$ to remove this constraint. Assuming that the upper $u\times u$ matrix of $\Gamma$ (say $\Gamma_1$) is invertible, we can express $\Gamma$ as
\begin{equation*}
\Gamma = \left( \begin{array}{c}{\Gamma_{1}} \\ {\Gamma_{2}}\end{array}\right) = \left( \begin{array}{c}{\mathrm{I}} \\ {\Gamma_{2}\Gamma_1^{-1}}\end{array}\right)\Gamma_1,
\end{equation*}
then we know $\text{span}(\Gamma)$ is the same as the span of $ \left( \begin{array}{c}{\mathrm{I}} \\ {\Gamma_{2}\Gamma_1^{-1}}\end{array}\right)$. Moreover, span of $ \left( \begin{array}{c} {-(\Gamma_{2}\Gamma_1^{-1})^{\T}} \\ {\mathrm{I}} \end{array}\right)$ is the same of $\text{span}(\Gamma_0)$. Therefore, by replacing $\eta$ now with $\Gamma_1\eta$ and $\Omega$ with $\Gamma_1\Omega\Gamma_1^{\T}$, we get a new parameterization where $\Gamma$ has the form of $\left( \begin{array}{c}{\mathrm{I}} \\ {A}\end{array}\right)$, $\Gamma_0$ has the form of $\left( \begin{array}{c}{-A^{\T}} \\ {\mathrm{I}} \end{array}\right)$, $\eta$, $\Omega$ and $\Omega_0$ are free parameter, and there is no constraint on $A$. We use this parametrization in the implementation of the enveloped Huber regression.

The optimization problem in \eqref{op} is non-convex in parameter $\zeta$. For that, we use the \texttt{fminsearch} function in the R package \texttt{neldermead}. The function uses Nelder-Mead method to find the minima of the objective function \citep{nm1965}. The initial value is crucial in non-convex optimization problems. We use the partial least squares estimator as an initial estimator in our algorithm. More details on the intimate connection between partial least squares and the envelope least squares model can be found in \cite{cook2013,cooksu2016, su2020}.

Notice that the dimension $u$ of the envelope space $\mathcal{E}_{\Sigma_{x}^*}\left(\beta^*\right)$ is also unknown in practice. To select the dimension $u$, we use the cross validation (or CV for short) method. Specifically, we randomly partition the data into $K$ groups of approximately equal size. Let $S_1, S_2, \dots, S_K$ be the sets which include the indices of data points in each group. For a fixed $u\in\{1,\dots,p\}$, let $\hat\mu^{(u,j)}$ and $\hat\beta^{(u,j)}$ be the enveloped Huber regression estimator for $\mu^*$ and $\beta^*$ using the data points whose indices are not in $S_j$. Then we choose $u\in\{1,\dots,p\}$ to minimize $CV(u) = \frac{1}{n}\sum_{j=1}^{K} \sum_{i \in S_j} \rho(y_i - \hat\mu^{(u,j)} - \bx_i^{\T}\hat\beta^{(u,j)})$. We treat $u$ as known in the asymptotic theory of the next section.

We have implemented EHR in R and the code is available upon request.

\section{Asymptotic Theory}\label{theory}
In this section we establish the asymptotic theory for the enveloped Huber regression model.
For notation convenience, let $\bw = (1,\bx^{\T})^{\T}$. For the purpose of comparison, we first provide a theorem that establishes the asymptotic normality of the joint Huber estimator, $\tilde\theta$, under a more general heteroscedastic error setting. 

\begin{theorem}\label{t2}
Assume $y_i = \mu^* + \bx_i^{\T} \beta^*+ \epsilon_i$, $\{(\bx_i,\epsilon_i)\}_{i=1}^{n}, (\bx,\epsilon)$ are i.i.d. with $\ev_{\theta^*}[\psi(\epsilon)|\bx]=0$, and $\bx \sim (\mu_x^*,\Sigma_x^*)$ with positive definite $\Sigma_x^*$ and finite fourth moments. Then
\begin{align}
&\sqrt{n} \begin{pmatrix}
   \hat\mu_h - \mu^*\\
   \hat\beta_h - \beta^*\\
  \vech(S_x) - \vech(\Sigma_x^*)\\
  \bbx - \mu_x
 \end{pmatrix} \overset{d}\to N\left(\mathbf{0},\bU_1^{-1}\bV_1\bU_1^{-1}\right),\nonumber
\end{align}
where $\bU_1$ is the block diagonal matrix $\bU_1 = (\bU_{1,ij})_{i,j=1,2,3}$ with $$\bU_{1,11} = \ev_{\theta^*}[\psi'(\epsilon)\bw\bw^{\T}], \;\bU_{1,22} = \mathrm{I}_{\frac{p(p+1)}{2}}, \;\bU_{1,33} = \mathrm{I}_{p},$$ and $\bV_1$ is the symmetric block matrix $\bV_1 = (\bV_{1,ij})_{i,j=1,2,3}$ with 
\begin{align}
&\bV_{1,11} = \ev_{\theta^*}[\psi^2(\epsilon)\bw\bw^{\T}], \;\bV_{1,22} = \var_{\theta^*}\Big(\vech\big((\bx-\mu_x^*)(\bx-\mu_x^*)^{\T}\big)\Big), \;\bV_{1,33} =  \var_{\theta^*}(\bx),\n\\ 
&\bV_{1,23} = \ev_{\theta^*}[\vech\big((\bx-\mu_x^*)(\bx-\mu_x^*)^{\T}\big) (\bx - \mu_x^*)^{\T}], \;\bV_{1,12} = 0, \;\bV_{1,13} = 0.\n
\end{align}

\end{theorem}

\begin{remark}\label{rmk2}
The proof of Theorem~\ref{t2} is relegated to section \ref{proof}. The standard asymptotic theory of the Huber regression is stated under the assumption that the error is independent of the covariates (for example, see chapter 5 of \cite{V1998}). See the asymptotic normality of the Huber regression given in remark~\ref{rmk1}. Theorem~\ref{t2} covers more cases where the error is allowed to be dependent on $\bx$. For example, a commonly used heteroscedastic regression model is
$
y_i = \ev(y_i|\bx_i)+\sigma(\bx_i)\tilde\epsilon_i,
$
where $\{\tilde\epsilon_i\}_{i=1}^n$ are i.i.d. with zero mean and independent from $\{\bx_i\}_{i=1}^n$, and $\sigma(\cdot)$ is a scale function. Then $\epsilon_i=\sigma(\bx_i)\tilde\epsilon_i$. When $\tilde\epsilon_i$ follows a symmetric distribution around zero,  $\ev_{\theta^*}[\psi(\epsilon_i) | \bx_i] =\ev_{\theta^*}[\psi(\sigma(\bx_i)\tilde\epsilon_i) | \bx_i] = 0$, which is seen by noticing $\sigma(\bx_i)\tilde\epsilon_i$ has the same conditional distribution as $\sigma(\bx_i)(-\tilde\epsilon_i)$ given $\bx_i$, and $\psi(t) = -\psi(-t)$ for any $t$. So, Theorem~\ref{t2} holds for this model when $\ev(y_i|\bx_i)$ is linear in $\bx_i$.

\end{remark}

For the asymptotic normality of the enveloped Huber regression estimator, $\hat\theta_{g}$, we assume the following regularity conditions
\begin{itemize}
\item[(i)] The support of $\zeta$ is compact, and so is the support $\Theta$ of $\theta$.  The true parameter 
$\theta^*$ is an interior point of $\Theta$. 
\item[(ii)] $\ev_{\theta^*}[\bg(\bz;\theta)]$ is differentiable at $\theta^*$ with $\frac{\partial \ev_{\theta^*}[\bg(\bz;\theta)]}{\partial \theta^{\T}}|_{\theta=\theta^*}$ having full rank and finite operator norm. The matrix $\ev_{\theta^*}[\bg(\bz;\theta^*)\bg(\bz;\theta^*)^{\T}]$ is positive definite and has finite operator norm. Among other implications, this condition implies that $\Sigma_x^*$ is positive definite and that $\bx$ has finite fourth moments (i.e. $\ev_{\theta^*}[\|\bx\|^4]<\infty$).
\item[(iii)] $y_i =\mu^*+ \bx_i^{\T} \beta^*+ \epsilon_i$, $\{(\bx_i,\epsilon_i)\}_{i=1}^{n}, (\bx,\epsilon)$ are i.i.d. with $\ev_{\theta^*}[\psi(\epsilon)|\bx]=0$. $\epsilon$ has continuous distribution and $\bx \sim (\mu_x^*,\Sigma_x^*)$. Besides, the envelope setup \eqref{setup} holds for some $\eta$ and $\Gamma$.
\end{itemize}

\begin{theorem}\label{t4}
Under the conditions $(i)-(iii)$, $\sqrt{n}(\hat\theta_{g} - \theta^*)$ converges to a multivariate normal distribution with mean zero and covariance matrix
\begin{equation*}
\mathrm{avar}(\sqrt{n}\hat\theta_{g}) = \Psi_1(\Psi_1^{\T}\mathrm{avar}(\sqrt{n}\tilde\theta)^{-1}\Psi_1)^{\dagger}\Psi_1^{\T}
\end{equation*}
where $\mathrm{avar}(\sqrt{n}\tilde\theta)$ is the covariance matrix of the limiting distribution given in Theorem \ref{t2}, and $\Psi_1 = \frac{\partial env(\zeta)}{\partial \zeta^{\T}}$ is the Jacobian matrix of $env(\cdot)$. In particular, a closed-form expression of $\Psi_1$ is 
$$\Psi_1 =  \begin{pmatrix}
  1 & 0 & 0 & 0 & 0 & 0\\
  0 &  \Gamma & \eta^{\T} \otimes \mathbf{I}_p & 0 & 0 & 0\\
  0 & 0 & 2C_p(\Gamma\Omega \otimes \mathbf{I}_p - \Gamma\otimes\Gamma_0\Omega_0\Gamma_0^{\T}) & C_p(\Gamma\otimes\Gamma)E_u & C_p(\Gamma_0 \otimes\Gamma_0)E_{p-u} & 0\\
  0 &  0 &  0 &  0 &  0 &  \mathbf{I}_p
 \end{pmatrix}.$$
\end{theorem}

There is apparently nothing in the envelope literature that is analogous to Theorem~\ref{t4} under normal theory because of the allowed dependence between $\bx$ and $\epsilon$.  The condition $\ev_{\theta^*}[\psi(\epsilon)|\bx]=0$ provides for some ``uncoupling'' of $\bx$ and $\epsilon$ but still allows dependence in higher-order moments. As discussed in 
Remark \ref{rmk2}, the EHR model and Theorem \ref{t4} can handle a commonly used heteroscedastic regression model
$
y_i = \ev(y_i|\bx_i)+\sigma(\bx_i)\tilde\epsilon_i,
$
where $\{\tilde\epsilon_i\}_{i=1}^n$ are i.i.d. and independent from $\{\bx_i\}_{i=1}^n$, and $\sigma(\cdot)$ is a scale function. Under such a setting, $\epsilon_i=\sigma(\bx_i)\tilde\epsilon_i$ and $\ev_{\theta^*}[\psi(\epsilon_i) | \bx_i] = 0$, as long as the distribution of $\tilde \epsilon_i$ is symmetric around zero.

\begin{corollary}\label{c1}
Under the same conditions in Theorem \ref{t4}, $\mathrm{avar}(\sqrt{n}\hat\theta_{g}) \preceq \mathrm{avar}(\sqrt{n}\tilde\theta)$.
\end{corollary}

\begin{corollary}\label{c2}
Under the same conditions in Theorem \ref{t4}, further assume that $\bx$ and $\epsilon$ are independent and $\ev_{\theta^*}[\bx] = 0$. Let $\hat\beta_{eh}$ be the part in $\hat\theta_g$ that estimates $\beta^*$, then we have
\begin{equation*}
\sqrt{n}(\hat\beta_{eh} - \beta^*) \overset{d}\to N\left(0,\frac{\ev_{\theta^*}[\psi^2(\epsilon)]}{(\ev_{\theta^*}[\psi'(\epsilon)])^2} \Gamma\Omega^{-1}\Gamma^{\T} + (\eta^{\T}\otimes \Gamma_0)T^{\dagger}(\eta\otimes\Gamma_0^{\T})\right),
\end{equation*}
where $T = \frac{(\ev_{\theta^*}[\psi'(\epsilon)])^2}{\ev_{\theta^*}[\psi^2(\epsilon)]} \eta\eta^{\T}\otimes \Omega_0 + \Omega\otimes \Omega_0^{-1} + \Omega^{-1}\otimes\Omega_0 - 2\mathbf{I}_u\otimes \mathbf{I}_{p-u}$. By Corollary \ref{c1}, we know that $\text{avar}(\sqrt{n}\hat\beta_{eh})\preceq \text{avar}(\sqrt{n}\hat\beta_h)$. Comparing $\text{avar}(\sqrt{n}\hat\beta_{eh})$ with $\text{avar}(\sqrt{n}\hat\beta_{eh}^*)$, we can view the term $ (\eta^{\T}\otimes \Gamma_0)T^{\dagger}(\eta\otimes\Gamma_0^{\T})$ as the cost in efficiency for not knowing $\Gamma$.
\end{corollary}

\citet[][Prop. 9]{cook2013} gave the asymptotic distribution of the maximum likelihood envelope estimator of $\beta$ when $y$ and $\bx$ are jointly normal.  The asymptotic distribution of Corollary~\ref{c2} is identical to their result except for the presence of the Huber factor $\ev_{\theta^*}[\psi^2(\epsilon)] / (\ev_{\theta^*}[\psi'(\epsilon)])^2$. This implies in part that the cost in efficiency for not knowing $\Gamma$ is the same for the normal maximum likelihood and Huber estimators. If $\var(\epsilon)$ is finite and we set the Huber factor equal to $ \var(\epsilon)$ then we get the asymptotic distribution of \cite{cook2013}.  They also showed that without normality the asymptotic distribution of the maximum likelihood envelope estimator of $\beta$ is still normal, but that its covariance matrix is complicated.  It seems remarkable that changing from maximum likelihood to Huber estimation mitigates the complexity of the asymptotic variance. It is also interesting to point out that the heteroscedastic error has another theoretical implication: the maximum likelihood envelope estimator in \cite{cook2013} may not be the most efficient even when the error is normal conditioned on $\bx$ because its asymptotic efficiency is based on the assumption that the errors are i.i.d. and are independent with $\bx$. The simulation results in section~\ref{sim2} and \ref{dim} confirm this point.

\section{Simulations}\label{simulation}
In this section, we demonstrate the estimation efficiency gains of the enveloped Huber regression (EHR) model through simulated experiments. The parameter $k$ in the Huber loss was chosen by the method described in section 2.2.

\subsection{The homoscedastic cases}\label{sim}
We consider the following simulation settings:
\begin{align}\label{homo}
y_i = 5 + \bx_i^{\T}\beta^* + \epsilon_i, i=1,\dots,n,
\end{align}
where $\bx_i,\beta^* \in \real^p$. We set $p = 12$, $u=2$, and all elements in $\beta^*$ were $0.1$.   
We fixed the sample size at $n=500$. 
The predictor $\bx$ was generated from a multivariate normal distribution with mean $\mathbf{0}$ and covariance matrix $\Sigma_x^* = \Gamma\Omega\Gamma^{\T} + \Gamma_0\Omega_0\Gamma_0^{\T}$, where $\Omega$ is a $u\times u$ diagonal matrix with diagonal elements 9 and 100, and $\Omega_0$ is the $p-u$ dimensional identity matrix. The matrix $\Gamma\in\real^{p\times u}$ is a semi-orthogonal matrix, with its $i$th row being $(-\frac{1}{\sqrt{\frac{p}{2}}},0)$ if $i$ is odd number, and the rest rows being $(0,-\frac{1}{\sqrt{\frac{p}{2}}})$. The matrix $\Gamma_0\in\real^{p\times (p-u)}$ is a semi-orthogonal matrix such that $\Gamma_0^{\T}\Gamma = 0$. The vector $\eta$ is set as $\eta=(-0.1 \sqrt{\frac{p}{2}},-0.1 \sqrt{\frac{p}{2}})^{\T}$ so that $\beta^* = \Gamma\eta$. Therefore, $\beta^*\in\text{span}(\Gamma) = \mathcal{E}_{\Sigma_{x}^*}\left(\beta^*\right)$.

We first consider the usual independent homoscedastic error cases. The simulation results under conditional heteroscedastic errors are presented in section \ref{sim2} and \ref{dim}. Six different error distributions were considered: 
\begin{itemize}
\item the standard normal distribution $N(0,1)$;
\item student's t-distribution with 3 degrees of freedom $t_3$;
\item the normal mixture $0.9N(0,1)+0.1N(0,25)$; 
\item the standard Laplace distribution $Laplace(0,1)$ with density function $\frac{1}{2}e^{-|x|},x\in\real$; 
\item a symmetric Gamma distribution $sGamma(2,2)$, which is the distribution of a random variable $ZV$ where $Z$ equals $1$ and $-1$ each with probability $\frac{1}{2}$, $V$ is independent of $Z$ and follows a Gamma distribution with shape and scale parameters both being 2. We consider such a distribution because it is bimodal;
\item the standard Cauchy distribution $Cauchy(0,1)$.
\end{itemize}

The asymptotic theory in section \ref{theory} (see Corollary \ref{c2}) indicates that the Huber factor $\frac{\ev_{\theta^*}[\psi^2(\epsilon)]}{(\ev_{\theta^*}[\psi'(\epsilon)])^2}$ plays an critical role in determining asymptotic efficiency of EHR. In comparison, $\var(\epsilon)$ is the corresponding term that affects the efficiency of standard envelope estimator (ENV). We provide the Huber factor and error variance for each of the six distributions as well as the ratio of these two factors in Table \ref{t_0}. Note that for Cauchy distribution the second moment does not exist so the corresponding $\var(\epsilon)$ is $\infty$.

\begin{table}[!h]
\centering
\caption{Huber factor and error variance under six different distributions.}
\begin{tabular}{@{}cccc@{}}
\toprule
 & $\frac{\ev_{\theta^*}[\psi^2(\epsilon)]}{(\ev_{\theta^*}[\psi'(\epsilon)])^2}$ &  $\var(\epsilon)$ 			& $\frac{\var(\epsilon)}{\ev_{\theta^*}[\psi^2(\epsilon)] / (\ev_{\theta^*}[\psi'(\epsilon)])^2}$      \\ \midrule
$N(0,1)$   & $1.05$ &   $1$ & $0.95$           \\
$t_3$       &  $1.59$  &  $3$    & $1.89$      \\ 
$0.9N(0,1)+0.1N(0,25)$   & $1.38 $         &  $3.4$     	     & $2.47$       \\ 
$Laplace(0,1)$     & $1.43$       & $2 $             & $1.39$       \\
$sGamma(2,2)$  & $23.82$ 	& $24$  &$1.01$ \\
$Cauchy(0,1)$ &$3.51 $ & $\infty $ & $\infty$  \\ \bottomrule
\end{tabular}
\label{t_0}
\end{table}

In this study, the envelope dimension $u$ was fixed at the true value. We generated $100$ independent datasets from the simulation model. For each replication, we computed the EHR estimator, the Huber regression (HR) estimator, the standard envelope estimator (ENV) which is based on normal likelihood, and the ordinary least squares estimator (LS). We compared the estimation accuracy of the four estimators. For an estimator generically denoted as $\hat\beta$, its estimation loss is defined as $(\hat\beta  - \beta^*)^{\T}(\hat\beta  - \beta^*)$. Table \ref{t_1} shows the estimation MSE averaged from the 100 replications. Table \ref{t_1} clearly demonstrates the efficiency gains achieved by the EHR estimator. When the error distribution is not normal, EHR substantially outperforms other estimators in terms of estimation. When the error distribution is normal, EHR still outperforms LS and HR, while maintaining a very low accuracy loss compared to ENV. Notice that when the error follows the Cauchy distribution, ENV and LS completely fail, while HR and EHR are still serviceable. 

\begin{table}[!h]
\centering
\caption{Comparison of estimation MSE ($\times 10^{-2}$) for simulation models described in (\ref{homo}). The results are based on 100 replications. The standard errors are listed in the parentheses ($\times 10^{-2}$). $u$ is fixed at the true value. $n$ is fixed at 500. ``mixnorm'' stands for the normal mixture $0.9N(0,1)+0.1N(0,25)$.}
\begin{tabular}{@{}ccccc@{}}
\toprule
$\epsilon$ & EHR &  	ENV 			& HR    & LS      \\ \midrule
 
$N(0,1)$   & $0.0377$ &   $0.0365$ & $2.25$   & $2.16$         \\
 &$(0.0023)$ &$(0.0023)$ &$(0.10)$ &$(0.10)$ \\
$t_3$       &  $0.050$  &  $0.08$    & $3.46$     & $6.06$          \\ 
&$(0.005)$&$(0.01)$&$(0.18)$&$(0.34)$\\
mixnorm   & $0.050 $        &  $0.10$     	     & $3.07$    & $7.35 $       \\ 
&$(0.004)$&$(0.01)$&$(0.16)$&$(0.38)$ \\
$Laplace(0,1)$     & $0.054$       & $0.07 $             & $3.08$   & $4.16$        \\
&$(0.004)$& $(0.01)$&$(0.13)$&$(0.19)$\\
$sGamma(2,2)$  & $0.63$ 	& $9.57$  &$51.58$ & $51.55$ \\
&$(0.07)$&$(2.38)$&$(2.19)$&$(2.34)$\\
$Cauchy(0,1)$ &$1.91 $ & $1.29\times 10^5 $ & $8.06$ & $2.02\times 10^5$ \\
&$(0.44)$&$(9.92\times 10^4)$&$(0.42)$&$(1.57\times 10^5)$\\ \bottomrule
\end{tabular}
\label{t_1}
\end{table}

\subsection{The conditional heteroscedastic cases}\label{sim2}
Our theory from section \ref{theory} demonstrates that the EHR estimator still enjoys asymptotic normality under conditional heteroscedastic errors. 
We conducted a set of simulations to confirm the asymptotic theory of EHR under conditional heteroscedasticity. We generated $n$ i.i.d. samples $\{(y_i,\bx_i)\}_{i=1}^n$ from the following heteroscedastic regression model 
\begin{align}\label{hmod}
y_i = \mu^*+\bx_i^{\T}\beta^*+\epsilon_i,  \epsilon_i = \sigma(\bx_i)\tilde\epsilon_i.
\end{align}
Here $\{\tilde\epsilon_i\}_{i=1}^n$ is independent of $\{\bx_i\}_{i=1}^n$ and the following three distributions of $\tilde\epsilon_i$ were considered: the standard normal distribution $N(0,1)$; student's t-distribution with 3 degrees of freedom $t_3$; and the normal mixture $0.9N(0,1)+0.1N(0,25)$. As before, we set $\mu^*=5$, $p=12$, $u=2$, and the settings for distribution of $\bx_i$ and $\beta^*$ were exactly the same as we introduced in section \ref{sim}. We set $\sigma(\bx) = x_1+x_{12}$ and $\sigma(\bx) = x_1*x_{12}$.
Similar linear scale functions have been considered in \cite{wang2012quantile,gu2016} which studies high dimensional linear heteroscedastic regression. Our theory can deal with both linear and nonlinear scale functions.

The envelope dimension $u$ was fixed at the true value. The sample size $n$ is 500. We compared the estimation MSE for the four estimators. From Table \ref{t_2a} and Table \ref{t_2b}, we can see that EHR substantially outperforms all its competitors under the case of conditional heteroscedasticity for all the six distributions. It is not surprising to see that ENV and LS fail when the distribution of $\tilde\epsilon$ is Cauchy. Moreover, we observe that even when $\tilde\epsilon$ follows a normal distribution, ENV does not perform  well compared with EHR. ENV which is based on normal likelihood \citep{cook2013} is the most efficient under the homoscedastic normal error assumption. However, our simulation results indicate that ENV is not the most efficient under conditional heteroscedastic normal errors.


\begin{table}[!h]
\centering
\caption{Comparison of estimation MSE for simulation model (\ref{hmod}) $y_i = \mu^*+\bx_i^{\T}\beta^*+\sigma(\bx_i)\tilde\epsilon_i$. The scale function is $\sigma(\bx) = x_1+x_{12}$. The results are based on 100 replications. The distributions of $\tilde\epsilon$ are listed in the first column. The standard errors are listed in the parentheses. $u$ is fixed at the true value. $n$ is 500. ``mixnorm'' stands for the normal mixture $0.9N(0,1)+0.1N(0,25)$.}
\begin{tabular}{@{}ccccc@{}}
\toprule
$\tilde\epsilon$ & EHR &  	ENV 			& HR    & LS      \\ \midrule
 
$N(0,1)$   & $0.005 $ 	&  $0.04 $    	& $0.20 $    	& $0.38 $          \\
 &$(0.001)$		&$(0.01)$		&$(0.01)$		&$(0.02)$ \\
$t_3$       & $0.006 $ 		&  $0.55 $    		& $0.34 $    	& $1.34 $         \\ 
&$(0.001)$		&$(0.09)$			&$(0.01)$		&$(0.09)$\\
mixnorm   & $0.007 $ 	&  $0.62 $    	& $0.27$    	& $1.40$       \\ 
&$(0.002)$	  	&$(0.09)$		&$ (0.02)$		&$ (0.08)$ \\ 
$Laplace(0,1)$     & $0.0040$       & $0.30$             & $0.26$   & $0.96$        \\
&$(0.0005)$& $(0.05)$&$(0.01)$&$(0.05)$\\
$sGamma(2,2)$  & $0.27$ 	& $5.86$  &$5.44$ & $9.48$ \\
&$(0.08)$&$(0.45)$&$(0.25)$&$(0.41)$\\
$Cauchy(0,1)$ &$0.021 $ & $3.10\times 10^4 $ & $0.67$ & $4.28\times 10^4$ \\
&$(0.003)$&$(2.05\times 10^4)$&$(0.03)$&$(2.73\times 10^4)$\\ \bottomrule
\end{tabular}
\label{t_2a}
\end{table}

\begin{table}[!h]
\centering
\caption{Comparison of estimation MSE for simulation model (\ref{hmod}) $y_i = \mu^*+\bx_i^{\T}\beta^*+\sigma(\bx_i)\tilde\epsilon_i$.  The scale function is $\sigma(\bx) = x_1*x_{12}$. The results are based on 100 replications. The distributions of $\tilde\epsilon$ are listed in the first column. The standard errors are listed in the parentheses. $u$ is fixed at the true value. $n$ is 500. ``mixnorm'' stands for the normal mixture $0.9N(0,1)+0.1N(0,25)$.}
\begin{tabular}{@{}ccccc@{}}
\toprule
$\tilde\epsilon$ & EHR &  	ENV 			& HR    & LS      \\ \midrule
 
$N(0,1)$   & $0.005$ 	&  $0.27$    	& $0.17$    	& $0.92$          \\
 &$(0.001)$		&$(0.04)$		&$(0.01)$		&$(0.05)$ \\
$t_3$       & $0.012$ 		&  $1.28$    		& $0.21$    	& $2.46$         \\ 
&$(0.002)$		&$(0.17)$			&$(0.01)$		&$(0.18)$\\
mixnorm   & $0.011$ 	&  $1.48$    	& $0.21$    	& $2.95$       \\ 
&$(0.003)$	  	&$(0.27)$		&$(0.01)$		&$(0.28)$ \\ 
$Laplace(0,1)$     & $0.010$       & $0.73$             & $0.20$   & $1.77$        \\
&$(0.003)$& $(0.09)$&$(0.01)$&$(0.08)$\\
$sGamma(2,2)$  & $0.12$ 	& $15.55$  &$5.26$ & $23.49$ \\
&$(0.02)$&$(1.00)$&$(0.24)$&$(1.14)$\\
$Cauchy(0,1)$ &$0.03$ & $6.49\times 10^4$ & $0.50$ & $8.24\times 10^4$ \\
&$(0.01)$&$(5.90\times 10^4)$&$(0.03)$&$(7.52\times 10^4)$\\ \bottomrule
\end{tabular}
\label{t_2b}
\end{table}

\subsection{Cross validation selection of the dimension $u$}\label{dim}

In practice, we use cross-validation (or CV for short) to select $u$. We now examine the selection performance of 5-fold CV. First, we consider the homoscedastic simulation settings described in section \ref{sim}. For demonstration purpose, we used the standard normal, $t_3$ and the mixture normal distribution as examples. 
Table \ref{t_3} shows the frequency that the true dimension ($u=2$) was selected by different methods in different cases. Overall, we see that CV works reasonably well as a selection method with finite sample sizes.

\begin{table}[!h]
\centering
\caption{Frequency of selecting the true $u$ for simulation models described in (\ref{homo}). The results are based on 100 replications. 5-fold CV was used to select $u$ in EHR and ENV. ``mixnorm'' stands for the normal mixture $0.9N(0,1)+0.1N(0,25)$.}
\begin{tabular}{@{}cccccc@{}}
\toprule
	$n$		&				&$\epsilon\sim N(0,1)$ &  $\epsilon\sim t_3$ 			 & $\epsilon\sim \text{mixnorm}$                     \\ \midrule
 	\multirow{2}{*}{100}	& EHR		&  $82\%$    			& $91\%$    		& $90\%$                  \\
					&ENV  				&  $85\%$    			& $79\%$    		& $82\%$              \\ \midrule
 \multirow{2}{*}{500}	& EHR						&  $94\%$    			& $93\%$    		& $92\%$                  \\
							&ENV  				&  $89\%$    			& $90\%$    		& $87\%$              \\  \bottomrule
\end{tabular}
\label{t_3}
\end{table}

We further examined the estimation accuracy of CV-tuned envelope Huber estimators under the same model setup. Using $N(0,1)$, $t_3$ and mixture normal distributions as examples, we compared the estimation MSE of the four estimators, with EHR and ENV using their selected $u$. Table \ref{t_5} shows the estimation MSE from which we can see that the CV-tuned EHR estimator is the overall winner.

\begin{table}[!h]
\centering
\caption{Comparison of estimation MSE ($\times 10^{-2}$) with $u$ selected by CV, in the case of homoscedastic error. The model is described in (\ref{homo}). The results are based on 100 replications. Distributions of the error $\epsilon$ are shown in the second column. The standard errors are listed in the parentheses ($\times 10^{-2}$). 5-fold CV was used to select $u$ in EHR and ENV. ``mixnorm'' stands for the normal mixture $0.9N(0,1)+0.1N(0,25)$.}
\begin{tabular}{@{}ccccccccc@{}}
\toprule
        $n$                  & $\epsilon$ & EHR & ENV & HR & LS \\ \midrule
\multirow{6}{*}{$100$} 
			& $N(0,1)$   & $2.57$ 	&  $2.43$      & $12.96$    	& $ 12.32$     \\
&&$(0.63)$&$(0.64)$  	&$(0.63)$&$(0.59)$\\
			& $t_3$   & $2.74 $ 	&  $5.45 $      & $21.38 $    	& $37.72 $     \\
&&$(1.00)$&$(1.36)$  	&$(1.22)$&$(3.02)$\\
                         & mixnorm  	& $2.01 $ 	&  $5.91 $       		& $17.42 $    	& $39.81 $ \\ 
                         &&$(0.65)$&$(1.96)$   	&$(0.92)$&$(2.37)$\\ \midrule
\multirow{6}{*}{$500$} 
			& $N(0,1)$   & $0.22$ 	&  $ 0.36$      & $2.07$    	& $1.94$     \\
&&$(0.07)$&$(0.10)$  	&$(0.07)$&$(0.07)$\\
			& $t_3$   & $0.41 $ 	&  $0.95 $       		& $3.46$     & $6.06$     \\
&&$(0.13)$&$(0.27)$  	&$(0.18)$&$(0.34)$\\
                         & mixnorm  	& $0.41 $ 	&  $1.66 $    		& $3.07$    & $7.35 $ \\
                         &&$(0.14)$&$(0.43)$  &$(0.16)$&$(0.38)$\\  \bottomrule
\end{tabular}
\label{t_5}
\end{table}

Next, we consider the heteroscedastic simulation settings described in section \ref{sim2}, i.e. model \eqref{hmod}. Again, we considered both linear and nonlinear scale functions: $\sigma(\bx) = x_1+x_{12}$ or $\sigma(\bx) = x_1* x_{12}$. In Table \ref{t_6.5} we present the frequency of the true $u$ being selected by 5-fold CV.  Similar to the homoscedastic cases, CV works very well for EHR as a selection method with finite sample size.
We also examined the estimation accuracy for EHR, ENV, HR and LS, where  EHR and ENV use the CV-tuned dimension. 
We summarize the estimation MSE for the four methods  in Table \ref{t_6.2} and Table~\ref{t_6.4}. It is clear that CV-tuned EHR significantly outperforms all the competing methods in all cases. 

\begin{table}[!h]
\centering
\caption{Frequency of selecting the true $u$ under model (\ref{hmod}) $y_i = \mu^*+\bx_i^{\T}\beta^*+\sigma(\bx_i)\tilde\epsilon_i$. The results are based on 100 replications. 5-fold CV was used to select $u$. Case \rn{1}: $\sigma(\bx) = x_1+x_{12}$; Case \rn{2}: $\sigma(\bx) = x_1*x_{12}$. ``mixnorm'' stands for the normal mixture $0.9N(0,1)+0.1N(0,25)$.}
\begin{tabular}{@{}cccccc@{}}
\toprule
Case	&	$n$		&				&$\tilde\epsilon\sim N(0,1)$ &  $\tilde\epsilon\sim t_3$ 			 & $\tilde\epsilon\sim \text{mixnorm}$                     \\ \midrule
\multirow{4}{*}{\rn{1}} 	&\multirow{2}{*}{100}	& EHR					&  $0.71$    			& $0.62$    		& $0.57$                  \\
		&					&ENV  				&  $0.13$    			& $0.11$    		& $0.15$              \\ \cmidrule{3-6}
& \multirow{2}{*}{500}	& EHR						&  $0.93$    			& $0.84$    		& $0.81$                  \\
&							&ENV  				&  $0.68$    			& $0.33$    		& $0.24$              \\ \midrule
\multirow{4}{*}{\rn{2}} 	&\multirow{2}{*}{100}	& EHR					&  $0.55$    			& $0.56$    		& $0.59$                  \\
		&					&ENV  				&  $0.13$    			& $0.09$    		& $0.10$              \\ \cmidrule{3-6}
& \multirow{2}{*}{500}	& EHR 						&  $0.86$    			& $0.83$    		& $0.83$                  \\
&							&ENV  				&  $0.33$    			& $0.24$    		& $0.12$              \\ \bottomrule
\end{tabular}
\label{t_6.5}
\end{table}

\begin{table}[!h]
\centering
\caption{Comparison of estimation MSE with $u$ selected by CV, in the case of conditional heteroscedastic error. The model is (\ref{hmod}) $y_i = \mu^*+\bx_i^{\T}\beta^*+\sigma(\bx_i)\tilde\epsilon_i$ with $\sigma(\bx) = x_1+x_{12}$. The distribution of $\tilde\epsilon$ is given in the third column. The results are based on $100$ replications. The standard errors are listed in the parentheses. $u$ is selected by 5-fold CV in EHR and ENV. ``mixnorm'' stands for the normal mixture $0.9N(0,1)+0.1N(0,25)$.}
\begin{tabular}{@{}cccccccccccccc@{}}
\toprule
          $n$               & $\tilde\epsilon$  & EHR & ENV  & HR & LS \\ \midrule
\multirow{6}{*}{$100$} 
			& $N(0,1)$   	& $0.13$			&  $0.83$	    	& $1.26$    	& $2.22$     \\
		 	&    	       		& $(0.05)$			&  $(0.16)$	    	& $(0.07)$    	& $(0.12)$     \\
			& $t_3$   		& $0.17$			&  $2.86$	    	& $1.90$    	& $6.22$     \\
		 	&    	       		& $(0.06)$			&  $(0.60)$	    	& $(0.12)$    	& $(0.50)$     \\
                       & mixnorm  	& $0.14$			&  $3.94$	    	& $1.86$    	& $8.49$     \\
                       &			& $(0.06)$			&  $(0.60)$	    	& $(0.11)$    	& $(0.83)$     \\ \midrule
\multirow{6}{*}{$500$} 
			& $N(0,1)$   		&$0.03$	&$0.13$	    	& $0.20 $    	& $0.38 $     \\
		 	&    	       			&$(0.01)$	&$(0.03)$		&$(0.01)$		&$(0.02)$\\
			& $t_3$   			&$0.05$	&$0.27$	    		& $0.34 $    	& $1.34 $     \\
			&				&$(0.01)$	&$(0.07)$			&$(0.01)$		&$(0.09)$\\
                       & mixnorm 			&$0.05$	&$0.29$	    	& $0.27$    	& $1.40$ \\ 
                       &				&$(0.01)$	&$(0.06)$	&$ (0.02)$		&$ (0.08)$\\ \bottomrule
\end{tabular}
\label{t_6.2}
\end{table}

\begin{table}[!h]
\centering
\caption{Comparison of estimation MSE with $u$ selected by CV, in the case of conditional heteroscedastic error. The model is (\ref{hmod}) $y_i = \mu^*+\bx_i^{\T}\beta^*+\sigma(\bx_i)\tilde\epsilon_i$  with $\sigma(\bx) = x_1*x_{12}$. The distribution of $\tilde\epsilon$ is given in the second column. The results are based on $100$ replications. The standard errors are listed in the parentheses. $u$ is selected by 5-fold CV in EHR and ENV. ``mixnorm'' stands for the normal mixture $0.9N(0,1)+0.1N(0,25)$.}
\begin{tabular}{@{}cccccccccccccc@{}}
\toprule
$n$                       & $\tilde\epsilon$  & EHR	& ENV  & HR & LS \\ \midrule
\multirow{6}{*}{$100$} 
			& $N(0,1)$   	& $0.11$			&  $1.44$	    	& $1.19$    	& $5.71$     \\
		 	&    	       		& $(0.02)$			&  $(0.27)$	    	& $(0.08)$    	& $(0.46)$     \\
			& $t_3$   		& $0.08$			&  $8.59$	    	& $1.77$    	& $15.65$     \\
		 	&    	       		& $(0.02)$			&  $(2.50)$	    	& $(0.12)$    	& $(4.59)$     \\
                       & mixnorm  	& $0.08$			&  $9.22$	    	& $1.74$    	& $17.81$     \\
                       &			& $(0.02)$			&  $(1.38)$	    	& $(0.12)$    	& $(2.00)$     \\ \midrule
\multirow{6}{*}{$500$} 
			& $N(0,1)$   	&$0.03$	&$0.20$	  	& $0.17 $    	& $0.92 $     \\
		 	&    	       		&$(0.01)$	&$(0.05)$		&$(0.01)$		&$(0.05)$\\
			& $t_3$   		&$0.04$	&$0.62$	   		& $0.21 $    	& $2.46 $     \\
			&			&$(0.01)$	&$(0.14)$		&$(0.01)$		&$(0.18)$\\
                       & mixnorm  	&$0.03$	&$0.96$	    	& $0.21$    	& $2.95$ \\ 
                       &			&$(0.01)$	&$(0.19)$		&$ (0.01)$		&$ (0.28)$\\ \bottomrule
\end{tabular}
\label{t_6.4}
\end{table}

 \section{A Real Data Example}

We applied the EHR model to analyze a real dataset, ``statex77'', which is contained in the \texttt{datasets} package in R. The dataset includes eight measurements for 50 states in United States, which are population, income, illiteracy, life expectancy, murder rate, percentage of high-school graduates, frost level and land area. We treated the murder rate as the response variable and treated the other variables as the predictors. A preliminary look at the data suggests that the murder rate has a heavy tail distribution, indicating the presence of outliers. Since the predictors have different measurement units, they were rescaled to have unit empirical standard deviation. We applied CV to choose the envelope dimension $u$ and then fitted the EHR model with the selected $u$. For comparison, we computed the ENV estimator with its selected $u$, the HR estimator and the LS estimators. Bootstrap was used to estimate the standard deviation of the four estimators. The ratios of bootstrapped standard deviation can be viewed as the measure for estimation efficiency comparison.   

\begin{table}[!h]
\centering
\caption{Analysis of ``statex77'' dataset using the EHR estimator, the HR estimator, the ENV estimator and the LS estimator. Columns 1 gives the dimension $u$ that is selected by CV in EHR model. Since $\beta\in\real^7$, we calculated the ratio of the bootstrapped standard deviations of other estimators versus the EHR estimator for each component in $\beta$. Columns 2-3 contain the range and the average of the bootstrap standard deviations ratios of the HR estimator versus the EHR estimator. Columns 4-5 contain the range and the average of the bootstrap standard deviations ratios of the ENV estimator versus the EHR estimator. Columns 6-7 contain the range and the average of the bootstrap standard deviations ratios of the LS estimator versus the EHR estimator.}
\begin{tabular}{cccccccccc}
\toprule
\multirow{2}{*}{$\hat u$} &  & \multicolumn{2}{c}{HR to EHR} &  & \multicolumn{2}{c}{ENV to EHR} &  & \multicolumn{2}{c}{LS to EHR} \\ \cline{3-4} \cline{6-7} \cline{9-10} 
                          &  & Range         & Average        &  & Range        & Average        &  & Range        & Average        \\ \midrule
          1                &  &     $1.74$--$3.91$          &  $2.68$              &  &     $1.24$--$3.42$         &      $2.15$          &  &         $1.50$--$3.44$     &     $2.38$         \\ \bottomrule
\end{tabular}
\label{t_7}
\end{table}

From Table \ref{t_7} we see massive efficiency gains of the EHR estimator over all the other estimators. CV selects $u=1$ in the EHR model. 
There are seven covariates in the dataset. For each component in $\beta$, we computed the ratio of the bootstrapped standard deviations of other estimators versus the EHR estimator, which are summarized in columns 2-7 of Table \ref{t_7}. The EHR method also reveals different findings from the other models. The significant predictors from different methods are summarized in Table \ref{t_8}. LS and HR method have two significant predictors while ENV has four. However, under the EHR model, five among seven predictors are significant. The EHR model suggests that illiteracy has significantly positive association with murder rate, while income, life expectancy, percentage of high-school graduates and frost (number of days with minimum temperature below freezing) have significantly negative association with murder rate. The difference between EHR and other models can be explained by the efficiency gain of the EHR estimator that allows it to detect weaker signals in the data. 

\begin{table}[!h]
\centering
\caption{Significant predictors in different methods}
\begin{tabular}{|c|c|}
\hline
 EHR		&  income, illiteracy, life expectancy, percentage of high-school graduates, frost level   \\ \hline
ENV  				& population, illiteracy, life expectancy, land area                  \\ \hline
HR				&population, life expectancy                   \\ \hline
LS 				& population, life expectancy      \\ \hline
\end{tabular}
\label{t_8}
\end{table}

\section*{Appendix: Proofs}\label{proof}

\begin{proof}[Proof of Proposition \ref{equivdef}]
Recall the model (\ref{Hubermodel}) states that we have $y = \mu^*+\bx^{\T}\beta^* + \epsilon$ with $\ev[\psi(\epsilon)|\bx]=0$. We prove the statements (i) and (ii) are equivalent.

(ii) $\implies$ (i): Since $\beta^*\in S$, we have $y = \mu^*+\bx^{\T}\bP_S\beta^*+\epsilon$. By definition of $\ev^{\rho}[\cdot]$, we have $\ev^{\rho}[y|\bP_S\bx] = \mu^* + \bx^{\T}\bP_S\beta^* + \ev^{\rho}[\epsilon|\bP_S\bx]$. Since $\ev[\psi(\epsilon)|\bx] = 0$, we have $\ev[\psi(\epsilon)|\bP_S\bx] = \ev\big[\ev[\psi(\epsilon)|\bx]|\bP_S\bx\big] = 0$, therefore $\ev^{\rho}[\epsilon|\bP_S\bx]=0$. So $\ev^{\rho}[y|\bP_S\bx] = \mu^* + \bx^{\T}\bP_S\beta^* = \mu^* + \bx^{\T}\beta^* = \ev^{\rho}[y|\bx]$.

(i) $\implies$ (ii): Since $\ev^{\rho}[y|\bx] = \ev^{\rho}[y|\bP_S\bx]$, we have $\mu^* + \bx^{\T}\bP_S\beta^* + \bx^{\T}\bQ_S\beta^* = \ev^{\rho}[(\mu^*+\bx^{\T}\bP_S\beta^*+\bx^{\T}\bQ_S\beta^*+\epsilon) | \bP_S\bx]$, which is equivalent to $\bx^{\T}\bQ_S\beta^* = \ev^{\rho}[(\bx^{\T}\bQ_S\beta^* + \epsilon)|\bP_S\bx]$. This implies that $\bx^{\T}\bQ_S\beta^*$ is a function of $\bP_S\bx$, i.e. $\bx^{\T}\bQ_S\beta^* = f(\bP_S\bx)$ for some fixed univariate function $f(\cdot)$. If $\bQ_S\beta^*\ne 0$, then consider the vector $\tilde\bx \coloneqq \bx+\bQ_S\beta^*$. We have $\bP_S\tilde\bx = \bP_S\bx + \bP_S\bQ_S\beta^* = \bP_S\bx$. On the other hand, $f(\bP_S\tilde\bx)= \tilde\bx^{\T}\bQ_S\beta^* = \bx^{\T}\bQ_S\beta^* + \|\bQ_S\beta^*\|^2 > \bx^{\T}\bQ_S\beta^* = f(\bP_S\bx) = f(\bP_S\tilde\bx)$, so we reach a contradiction. Therefore $\bQ_S\beta^*=0$, which means $\beta^*\in S$.
\end{proof}

\begin{proof}[Proof of Theorem \ref{t2}]
For ease of notation, we denote $\bw_i = (1,\bx_i^{\T})^{\T}$, $\bw = (1,\bx^{\T})^{\T}$, and $G(\theta) = \ev_{\theta^*}[\bg(\bz;\theta)]$. To derive the asymptotic distribution of $\tilde\theta$, we will apply Theorem 3.3 of \cite{pakes1989}. To apply the theorem we need to verify its conditions (i)-(v). The condition (i) is satisfied since $G_n(\tilde\theta) = \mathbf{0}$. The conditions (ii) and (v) automatically holds given our condition (i). By our condition (ii) and Central Limit Theorem, we have $\sqrt{n}(G_n(\theta^*) - G(\theta^*)) \overset{d}\to N(0,\ev_{\theta^*}[\bg(\bz;\theta^*)\bg(\bz;\theta^*)^{\T}])$, thus the condition (iv) holds since $G(\theta^*) = \mathbf{0}$. By our Lemma \ref{l1}, we have $$\sup _{\theta :\left\|\theta-\theta^*\right\| \leq \delta_{n}}\frac{\left\|G_{n}(\theta)-G(\theta)-G_{n}\left(\theta^*\right)\right\|}{n^{-1/2} + \|G_{n}(\theta)\| + \|G(\theta)\|} = o_p(1)$$ for every sequence $\{\delta_n\}$ of positive numbers that converges to zero. Thus the condition (iii) holds. Let $\bw_i = (1,\bx_i^{\T})^{\T}$. Therefore, we have finished verifying conditions (i)-(v) in Theorem 3.3 of \cite{pakes1989}. In addition, it is easy to show $\tilde\theta\overset{P}\to \theta^*$, for example, through an argument similar to the proof our Lemma \ref{l2}. Now, applying their Theorem 3.3 we have $\sqrt{n}(\tilde\theta - \theta^*) \overset{d}\to N(\mathbf{0},(\bU_1^{\T}\bU_1)^{-1}\bU_1^{\T}\bV_1\bU_1(\bU_1^{\T}\bU_1)^{-1})$, where $\bU_1 = \frac{\partial \ev_{\theta^*}[\bg(\bz;\theta)]}{\partial \theta^{\T}}|_{\theta=\theta^*}$, and $\bV_1 = \ev_{\theta^*}[\bg(\bz;\theta^*)\bg(\bz;\theta^*)^{\T}]$. Recall that $\bg_1(\bz;\theta) = \psi(y - \mu - \bx^{\T}\beta)(1,\bx^{\T})^{\T}$. Since $\psi(\cdot)$ has a.e. derivative that is bounded by 1, we have $\frac{\partial \ev_{\theta^*}[\bg_1(\bz;\theta)]}{\partial (\mu, \beta^{\T})}|_{\theta=\theta^*} = \ev_{\theta^*}[\psi'(\epsilon)\bw\bw^{\T}]$. Then it is easy to give the expression of $\bU_1$ as $\bU_1 = \begin{pmatrix}
  \ev_{\theta^*}[\psi'(\epsilon)\bw\bw^{\T}] & 0 & 0 \\
  0 &  \mathrm{I}_{\frac{p(p+1)}{2}} & 0\\
  0 & 0 & \mathrm{I}_{p} 
 \end{pmatrix}$.
Next we give the expression of $\bV_1 = (\bV_{1,ij})_{i,j=1,2,3}$. It is easy to check
\begin{align}
\bV_{1,11} &= \ev_{\theta^*}[\bg_1(\bz;\theta^*)\bg_1(\bz;\theta^*)^{\T}] = \ev_{\theta^*}[\psi^2(\epsilon)\bw\bw^{\T}]\n\\
\bV_{1,22} &= \ev_{\theta^*}[\bg_2(\bz;\theta^*)\bg_2(\bz;\theta^*)^{\T}] =  \var_{\theta^*}\Big(\vech\big((\bx-\mu_x^*)(\bx-\mu_x^*)^{\T}\big)\Big)\n\\
\bV_{1,23} &= \ev_{\theta^*}[\bg_2(\bz;\theta^*)\bg_3(\bz;\theta^*)^{\T}] = \ev_{\theta^*}[\vech\big((\bx-\mu_x^*)(\bx-\mu_x^*)^{\T}\big) (\bx - \mu_x^*)^{\T}] \n\\
\bV_{1,33} &= \ev_{\theta^*}[\bg_3(\bz;\theta^*)\bg_3(\bz;\theta^*)^{\T}] = \var_{\theta^*}(\bx).\n
\end{align}
Also, $\bV_{1,12} = 0$ and $\bV_{1,13} = 0$ because $\ev_{\theta^*}[\psi(y - \mu^* - \bx^{\T}\beta^*)|\bx] = 0$. Notice that $\bU_1$ is symmetric, and it has full rank because of our condition (ii), so we have $\sqrt{n}(\tilde\theta - \theta^*)\overset{d}\to N(\mathbf{0},\bU_1^{-1}\bV_1\bU_1^{-1})$. So we complete the proof of Theorem \ref{t2}.
\end{proof}


The proof of Proposition \ref{p1} relies on the standard asymptotic result for Huber regression, which we restate here as Proposition \ref{c0}.
\begin{proposition}\label{c0}
Assume $y_i = \mu^* + \bx_i^{\T} \beta^*+ \epsilon_i$, $\{(\bx_i,\epsilon_i)\}_{i=1}^{n}, (\bx,\epsilon)$ are i.i.d. with $\ev_{\theta^*}[\psi(\epsilon)|\bx]=0$, and $\bx \sim (\mu_x^*,\Sigma_x^*)$ with positive definite $\Sigma_x^*$. If further assume that $\epsilon$ and $\bx$ are independent, then
\begin{equation*}
\sqrt{n} (\hat\beta_h - \beta^*) \overset{d}\to N(0,\frac{\ev_{\theta^*}[\psi^2(\epsilon)]}{(\ev_{\theta^*}[\psi'(\epsilon)])^2}\Sigma_x^{*-1}).
\end{equation*}
\end{proposition}

\begin{proof}[Proof of Proposition \ref{c0}]
By Theorem \ref{t2} and the conditions of the corollary, we have 
\begin{align}
&\sqrt{n} \begin{pmatrix}
   \hat\mu_h - \mu^*\\
   \hat\beta_h - \beta^*
 \end{pmatrix} \overset{d}\to N\left(\mathbf{0},\frac{\ev_{\theta^*}[\psi^2(\theta)]}{(\ev_{\theta^*}[\psi'(\theta)])^2}\ev_{\theta^*}[\bw\bw^{\T}]^{-1}\right).\nonumber
\end{align}
After computing the corresponding Schur complement in $\ev_{\theta^*}[\bw\bw^{\T}]$, we have $$\sqrt{n}(\hat\beta_h - \beta^*)\overset{d}\to N\left(\mathbf{0},\frac{\ev_{\theta^*}[\psi^2(\theta)]}{(\ev_{\theta^*}[\psi'(\theta)])^2} (\ev_{\theta^*}[\bx\bx^{\T}] - \ev_{\theta^*}[\bx]\ev_{\theta^*}[\bx^{\T}])\right) = N\left(\mathbf{0},\frac{\ev_{\theta^*}[\psi^2(\theta)]}{(\ev_{\theta^*}[\psi'(\theta)])^2}\Sigma_x^{*-1}\right).$$ So we finish the proof.
\end{proof}

\begin{proof}[Proof of Proposition \ref{p1}]
By Proposition \ref{c0} we know that
\begin{equation*}
\sqrt{n}(\hat\eta - \eta) \to  N(0,\frac{\ev_{\theta^*}[\psi^2(\epsilon)]}{(\ev_{\theta^*}[\psi'(\epsilon)])^2} (\text{var}(\Gamma^{\T}\bx))^{-1}).
\end{equation*}
where $\text{var}(\Gamma^{\T}\bx) = \Gamma^{\T}\Sigma_x^*\Gamma = \Omega$ by the relation in \eqref{setup}. So $\sqrt{n}(\hat\beta_{eh}^*-\beta^*) = \sqrt{n}\Gamma(\hat\eta-\eta)\to N(0,\frac{\ev_{\theta^*}[\psi^2(\epsilon)]}{(\ev_{\theta^*}[\psi'(\epsilon)])^2}\Gamma\Omega^{-1}\Gamma^{\T}) $.
\end{proof}

To prove Theorem \ref{t4}, we need the following lemmas \ref{l1}-\ref{l4}.  Before introducing these lemmas, let us introduce some notations first. For notational simplicity, let $L_{n}(\theta)=G_{n}^{\T}(\theta) \widehat{\boldsymbol{\Delta}} G_{n}(\theta)$ and $L(\theta)=G^{\T}(\theta) \boldsymbol{\Delta} G(\theta)$, where $G_{n}(\theta)$ is given by \eqref{hn}, $\widehat{\boldsymbol{\Delta}}=\left\{n^{-1} \sum_{i=1}^{n} \bg\left(\mathbf{Z}_{i}; \tilde{\theta}\right) \bg^{\T}\left(\mathbf{Z}_{i} ; \tilde{\theta}\right)\right\}^{-1}$, $G(\theta)=\mathrm{E}_{\theta^*}[\bg(\mathbf{Z}; \theta)]$, and $\Delta=\bV_1^{-1}=\left\{\mathrm{E}_{\theta^*}\left[\bg\left(\mathbf{Z} ; \theta^*\right) \bg\left(\mathbf{Z} ; \theta^*\right)^{\T}\right]\right\}^{-1}$ . Let $ q_{n}(\boldsymbol{\gamma})=G_{n}\left(\boldsymbol{\gamma} / \sqrt{n}+\theta^*\right)$ and $q(\boldsymbol{\gamma})=G(\boldsymbol{\gamma} / \sqrt{n}+\theta^*)$. Thus,  $q_{n}(0)=G_{n}\left(\theta^*\right)$ and $q(0)=G\left(\theta^*\right)=0$. Let $W_{n}(\boldsymbol{\gamma})=q_{n}^{\T}(\boldsymbol{\gamma}) \widehat{\boldsymbol{\Delta}} q_{n}(\boldsymbol{\gamma})$ and $W(\gamma)=q^{\T}(\gamma) \Delta q(\gamma)$. In addition, let $\tau_{n}(\gamma)=\left[q_{n}(\gamma)-q_{n}(0)-q(\gamma)\right] /(1+\|\gamma\|)$, $\phi_{n}(\gamma)=\tau_{n}^{\T}(\gamma) \widehat{\Delta} \tau_{n}(\gamma)+2 q_{n}^{\T}(0) \widehat{\Delta} \tau_{n}(\gamma)$, and $\varsigma_{n}(\gamma)=n\left[W_{n}(\gamma)-\phi_{n}(\gamma)-W_{n}(0)-\widehat{\mathbf{D}}^{\T} \gamma / \sqrt{n}-W(\gamma)\right]$, where $\widehat{\mathbf{D}}=2 \bU_1 \widehat{\Delta} q_{n}(0)$. Here, $\bU_1$ and $\bV_1$ are the same as in the proof of Theorem \ref{t2}, i.e. $\bU_1 = \frac{\partial \ev_{\theta^*}[\bg(\bz;\theta)]}{\partial \theta^{\T}}|_{\theta=\theta^*} = \begin{pmatrix}
  \ev_{\theta^*}[\psi'(\epsilon)\bw\bw^{\T}] & 0 & 0 \\
  0 &  \mathrm{I}_{\frac{p(p+1)}{2}} & 0\\
  0 & 0 & \mathrm{I}_{p} 
 \end{pmatrix}$, and $\bV_1 = \begin{pmatrix}
  \ev_{\theta^*}[\psi^2(\epsilon)\bw\bw^{\T}] & 0 & 0 \\
  0 &   \bV_{1,22} & \bV_{1,23} \\
 0 &  \bV_{1,23}^{\T} & \bV_{1,33}
 \end{pmatrix}$, where \newline $\bV_{1,22} = \var_{\theta^*}\Big(\vech\big((\bx-\mu_x^*)(\bx-\mu_x^*)^{\T}\big)\Big)$,
$\bV_{1,23} = \ev_{\theta^*}[\vech\big((\bx-\mu_x^*)(\bx-\mu_x^*)^{\T}\big) (\bx - \mu_x^*)^{\T}]$ and $\bV_{1,33} = \var_{\theta^*}(\bx)$.
\begin{lemma}\label{l1}
Under condition (iii), for every sequence of positive numbers $\delta_n = o(1)$, we have
\begin{align}\label{a1}
\sup _{\theta :\left\|\theta-\theta^*\right\| \leq \delta_{n}}\left\|G_{n}(\theta)-G(\theta)-G_{n}\left(\theta^*\right)\right\|=o_{p}\left(n^{-1 / 2}\right)
\end{align}
\end{lemma}
\begin{proof}[Proof of Lemma \ref{l1}]
For any $\theta$ and $\theta'$ such that $ \| \theta-\theta'\| \leq \delta_{n}$, we have
\begin{align}
&\left\|\bg_{1}\left(\bz ; \theta\right)-\bg_{1}\left(\bz; \theta'\right)\right\|^{2} \n\\
&\leq |\psi(y-\mu - \bx^{\T}\beta) - \psi(y -\mu' -  \bx^{\T}\beta')|^2\|\bw\|^2 \nonumber\\
&\le  2k |(\mu - \mu', \beta^{\T} - \beta^{'\T})\bw|\cdot \|\bw\|^2 \;\;\;\text{ (since $\psi(\cdot)$ is 1-Lipschitz and bounded by $k$)}\nonumber\\
&\le 2k\|\bw\|^3 \delta_n,\nonumber
\end{align}
therefore, we have
\begin{align}\label{lemma1_1}
&\ev_{\theta^*}\left[\sup_{\theta : \| \theta-\theta'\|  \leq \delta_{n}}\left\|\bg_{1}\left(\bz ; \theta\right)-\bg_{1}\left(\bz ; \theta'\right)\right\|^{2}\right]\le 2k\ev_{\theta^*}[\|\bw\|^3]\delta_n \triangleq c_1 \delta_n,
\end{align}
here $c_1$ is a positive constant.

Next, for any $\theta$ and $\theta'$ such that $ \| \theta-\theta'\| \leq \delta_{n}$, we have 
\begin{align}
&\left\|\bg_{2}\left(\bz; \theta\right)-\bg_{2}\left(\bz ; \theta'\right)\right\|^{2} \n\\
&= \Big\|\vech(\Sigma_x) - \vech\big((\bx - \mu_x)(\bx - \mu_x)^{\T}\big) -\vech(\Sigma_x') + \vech\big((\bx - \mu_x')(\bx - \mu_x')^{\T}\big)\Big\|^2 \n\\
&\le 2\delta_n^2 + 2\Big\|\vech\big((\bx - \mu_x)(\bx - \mu_x)^{\T}\big) - \vech\big((\bx - \mu_x')(\bx - \mu_x')^{\T}\big)\Big\|^2\n\\
&\le 2\delta_n^2 + 4\Big\|\vech\big((\bx - \mu_x')(\mu_x' - \mu_x)^{\T} +(\mu_x' - \mu_x)(\bx - \mu_x')^{\T} \big)\Big\|^2 + 4\Big\|\vech\big((\mu_x' - \mu_x)(\mu_x' - \mu_x)^{\T}\big)\Big\|^2\n\\
&\le 2\delta_n^2 + 16\|\bx-\mu_x'\|^2 \|\mu_x' - \mu_x\|^2 + 4\|\mu_x' - \mu_x\|^4\n\\
&\le 2\delta_n^2 + 4\delta_n^4 + 16\delta_n^2 \|\bx-\mu_x'\|^2\n\\
&\le 2\delta_n^2 + 4\delta_n^4 + 32\delta_n^2 \|\bx\|^2 + 32\delta_n^2\|\mu_x'\|^2.\n
\end{align}
Therefore, by the fact that $\ev_{\theta^*}[\|\bx\|^4]<\infty$, $\Theta$ is compact and $\delta_n = o(1)$, there exists some fixed positive constant $c_2$ such that
\begin{align}\label{lemma1_2}
\ev_{\theta^*}\left[\underset{\theta : \| \theta-\theta'\| \leq \delta_{n}}{\sup }\left\|\bg_{2}\left(\bz; \theta\right)-\bg_{2}\left(\bz ; \theta'\right)\right\|^{2}\right] \leq  c_2\delta_{n}^{2}.
\end{align}

For $\bg_3(\bz;\theta)$, we have $\left\|\bg_{3}\left(\bz; \theta\right)-\bg_{3}\left(\bz ; \theta'\right)\right\|^{2} = \|\mu_x - \mu_x'\|^2$, so
\begin{align}\label{lemma1_3}
\ev_{\theta^*}\left[\underset{\theta : \| \theta-\theta'\| \leq \delta_{n}}{\sup }\left\|\bg_{3}\left(\bz; \theta\right)-\bg_{3}\left(\bz ; \theta'\right)\right\|^{2}\right] \leq  \delta_{n}^{2}.
\end{align}
The results in \eqref{lemma1_1}, \eqref{lemma1_2} and \eqref{lemma1_3} together imply that $\bg(\bz; \theta)$ belongs to the ``type IV class'' of \cite{andrew1994} and is $L_2(P)$-continuous at $\theta$, for all $\theta\in\Theta$. For details regarding this statement, see (5.3) in \cite{andrew1994}. Thus, by applying Lemma 2.17 in \cite{pakes1989}, we have
\begin{equation*}
n^{-1 / 2} \sup _{\theta :\left\|\theta-\theta^*\right\| \leq \delta_{n}} \| \sum_{i=1}^{n}\left\{\bg(\bz_i ; \theta)-\mathrm{E}_{\theta^*}[\bg(\bz ; \theta)]-\bg\left(\bz_i ; \theta^*\right)\right\}| |=o_{p}(1).
\end{equation*}
Thus the lemma is proved.
\end{proof}

\begin{lemma}\label{l2}
Under the same condition as in Theorem \ref{t4}, $\widehat{\theta}_{g} \stackrel{p}{\longrightarrow} \theta^*$.
\end{lemma}
\begin{proof}[Proof of Lemma \ref{l2}]
Let $\mathcal{F}=\{\bg(\bz; \theta), \theta \in \Theta\}$, where $\Theta$  is compact under condition (i). By an argument similar to the proof of Lemma \ref{l1} and our condition (i), it is easy to verify that the condition of Lemma 2.13 in \cite{pakes1989} holds for each component of $\mathcal{F}$, and therefore each component of $\mathcal{F}$ (as a class of functions of $\bz$) is ``Euclidean'' for an envelope function. By our condition (iii) and expression of the envelope function given in the Lemma 2.13, it can be easily verified that all the envelope functions in our problem are integrable with respect to $\mathrm{P}_{\theta^*}$. Then by Lemma 2.8 of \cite{pakes1989}, we know the Uniform Law of Large Numbers holds for each component of $\mathcal{F}$. This implies
\begin{equation*}
\sup _{\theta \in \Theta}\left\|G_{n}(\theta)-G(\theta)\right\| \rightarrow 0, \text { a.s. }
\end{equation*}
As a result, $L_{n}(\theta)$ converges uniformly to $L(\theta)$ over $\Theta$ in probability. This implies that $L_{n}(\theta)$ converges uniformly to $L(\theta)$ in probability over a subset $\Theta_{e}$, where $\Theta_{e}=\{\theta : \theta \in \Theta \text{ and\;} \theta= env\left(\zeta\right) \text{ for some\;} \zeta\} $. Since the support of $\zeta$ is compact and $env(\cdot)$ is a continuous map, $\Theta_{e}$  is compact. Moreover, since $G\left(\theta^*\right)=0$  and $\theta^*$  is the unique root, $\theta^*$ is the unique minimizer of $L(\theta) $. By our condition (ii), $L(\theta)$ is continuous. Therefore, all the conditions of Theorem 2.1 in \cite{newey1994} hold, by the result of the Theorem 2.1 we have $\widehat{\theta}_{g} \stackrel{p}{\longrightarrow} \theta^*$.
\end{proof}

\begin{lemma}\label{l3}
Under the same conditions in Theorem \ref{t4}, for every sequence of positive numbers $\delta_n = o(1)$,
\begin{align}
\sup _{\vartheta} \frac{\left|\varsigma_{n}(\gamma)\right|}{\|\gamma\|(1+\|\gamma\|)}=o_{p}(1)\nonumber
\end{align}
where $\vartheta=\left\{\gamma :\|\gamma\| / \sqrt{n} \leq \delta_{n}\right\}$.
\end{lemma}
\begin{proof}[Proof of lemma \ref{l3}]
Based on the definition of $\tau_n(\gamma)$, we have
\begin{align}
W_{n}(\gamma)&=(1+ \|\gamma\|)^{2} \tau_{n}^{\T}(\gamma) \widehat{\Delta} \tau_{n}(\gamma)+q_{n}^{\T}(0) \widehat{\Delta} q_{n}(0)+q^{\T}(\gamma) \widehat{\Delta} q(\gamma)+2(1+\|\gamma\|) \tau_{n}^{\T}(\gamma) \widehat{\Delta} q_{n}(0)\nonumber\\
&+2(1+\|\gamma\|) \tau_{n}^{\T}(\gamma) \widehat{\Delta} q(\gamma)+2 q_{n}^{\T}(0) \widehat{\Delta} q(\gamma) \nonumber
\end{align}
and $W_{n}(0)=q_{n}^{\T}(0) \widehat{\Delta} q_{n}(0)$. Consequently, it can be shown that $\left|\varsigma_{n}(\gamma)\right| /[\|\gamma\|(1+\|\gamma\|)]\le\sum_{j=1}^{5} R_{j}(\gamma)$, where 
\begin{align}
&R_{1}(\gamma)=n(\|\gamma\|+2) \tau_{n}^{\T}(\gamma) \widehat{\Delta} \tau_{n}(\gamma) /(1+\| \gamma\|), \quad R_{2}(\gamma)=2 n\left|\tau_{n}^{\T}(\gamma) \widehat{\Delta} q_{n}(0)\right| /(1+\|\gamma\|), \nonumber\\
&R_{3}(\gamma)=2 n\left|\tau_{n}^{\T}(\gamma) \widehat{\Delta} q(\gamma)\right| /\|\gamma\|, R_{4}(\gamma)=n\left|2 q_{n}^{\T}(0) \widehat{\Delta} q(\gamma)-\widehat{\mathbf{D}}^{\T} \gamma / \sqrt{n}\right| /[\|\gamma\|(1+\|\gamma\|)],\nonumber\\
&R_{5}(\gamma)=n\left|q^{\T}(\gamma)(\widehat{\Delta}-\Delta) q(\gamma)\right| /[\|\gamma\|(1+\|\gamma\|)].\nonumber
\end{align}
To prove Lemma \ref{l3}, it suffices to show that $\sup_{\vartheta} R_{j}(\gamma)=o_{p}(1)$ for all $j=1,\dots,5$. From Lemma \ref{l1}, we know that $\sup _{\vartheta}\left\|\tau_{n}(\gamma)\right\|=o_{p}\left(n^{-1 / 2}\right)$. Then under condition (ii), $$\sup _{\vartheta} R_{1}(\gamma) \leq \sup _{\vartheta} n\left\|\tau_{n}(\gamma)\right\|^{2}\|\widehat{\Delta}\|(\|\gamma\|+2) /(1+\|\gamma\|)=n \sup _{\vartheta}\left\|\tau_{n}(\gamma)\right\|^{2} O_{p}(1)=o_{p}(1),$$ and $\sup _{\vartheta} R_{2}(\gamma) \leq \sqrt{n} \sup _{\vartheta}\left\|\tau_{n}(\gamma)\right\| O_{p}(1)=o_{p}(1)$. By Taylor expansion, $q(\gamma)=\bU_1 \gamma / \sqrt{n}+o(\gamma / \sqrt{n})$. So we have $$\sup _{\vartheta} R_{3}(\gamma) \leq 2 \sqrt{n} \sup _{\vartheta}\left\|\tau_{n}(\gamma)\right\|\|\widehat{\Delta}\|\left(\left\|\bU_1\right\|\|\gamma\|+o(\|\gamma\|)\right)/\|\gamma\|  \le \sqrt{n} \sup _{\vartheta}\left\|\tau_{n}(\gamma)\right\| O_{p}(1)=o_{p}(1),$$ and $\sup _{\vartheta} R_{4}(\gamma)=2 n \sup _{\vartheta} \left| q_{n}^{\T}(0) \widehat{\Delta}[q(\gamma) - \frac{\bU_1 \gamma}{\sqrt{n}}]\right| /[\|\gamma\|(1+\|\gamma\| )]  \leq \sqrt{n}\left\|q_{n}(0)\right\|\|\widehat{\Delta}\| o_{p}(1)=o_{p}(1)$. Finally, since $\sqrt{n}\|q(\boldsymbol{\gamma})\| \leq\|\gamma\|[\| \bU_1 \|+o(1)]$, $\sup _{\vartheta} R_{5}(\gamma) \leq \sup _{\vartheta} n\|q(\gamma)\|^{2}\|\widehat{\Delta}-\Delta\| /[\|\gamma\|(1+\|\gamma\|)] \leq \sup _{\vartheta}\|\widehat{\Delta}-\Delta\| O_{p}(1)=o_{p}(1)$.
\end{proof}
Note that $W_{n}(\gamma)$ is minimized at $\widehat{\gamma}_{g}=\sqrt{n}\left(\widehat{\theta}_{g}-\theta^*\right)$ under enveloping. 
\begin{lemma}\label{l4}
Under the same conditions in Theorem \ref{t4}, 
\begin{align}
\left\|\widehat{\gamma}_{g}\right\|=\sqrt{n}\|\widehat{\theta}_{g}-\theta^*\|=O_{p}(1).\nonumber
\end{align}
\end{lemma}
\begin{proof}[Proof of Lemma \ref{l4}]
Let $\vartheta$ be the same defined as in Lemma \ref{l3}. First, by the proof of Lemma \ref{l3},
\begin{align}
\sup _{\vartheta}\left|\phi_{n}(\gamma)\right| \leq O_{p}(1) \sup _{\vartheta}\left(\left\|\tau_{n}(\gamma)\right\|^{2}+2\left\|\tau_{n}(\gamma)\right\|\left\| q_{n}(0)\right\|\right) \leq o_{p}\left(n^{-1}\right) \nonumber
\end{align}
Note that under the envelope setting, $W_{n}\left(\widehat{\gamma}_{g}\right) \leq W_{n}(0)$, and by Lemma \ref{l2}, $\widehat{\gamma}_{g} \in \vartheta$ with probability going to 1, for some $\{\delta_n\}$ converging to zero. Hence $W_{n}\left(\widehat{\gamma}_{g}\right)-\phi_{n}(\gamma) \leq W_{n}\left(\widehat{\gamma}_{g}\right)+o_{p}\left(n^{-1}\right) \leq W_{n}(0)+o_{p}\left(n^{-1}\right)$. Therefore, define $\xi_n$ as
\begin{align}
\xi_n=-n\left[W_{n}\left(\widehat{\gamma}_{g}\right)-\phi_{n}(\gamma)-W_{n}(0)-o_{p}\left(n^{-1}\right)\right]=-\varsigma_{n}\left(\widehat{\gamma}_{g}\right)-\sqrt{n} \widehat{\mathbf{D}}^{\T} \widehat{\gamma}_{g}-n W\left(\widehat{\gamma}_{g}\right)+o_{p}(1), \nonumber
\end{align}
then with probability going to 1, we have $\liminf_{n\to\infty} \xi_n \ge 0$.

By Taylor expansion, $W\left(\widehat{\gamma}_{g}\right)=\widehat{\gamma}_{g}^{\T} \mathbf{H} \widehat{\gamma}_{g} /(2 n)+o\left(\left\|\widehat{\gamma}_{g}\right\|^{2} / n\right)$, where $$\mathbf{H}=n\left.\frac{\partial^{2} W(\boldsymbol{\gamma})}{\partial \boldsymbol{\gamma} \boldsymbol{\gamma}^{\T}}\right|_{\boldsymbol{\gamma}=0}=\left.\frac{\partial^{2} L(\theta)}{\partial \theta \theta^{\T}}\right|_{\theta=\theta^*}=2 \bU_1 \Delta \bU_1=2 \bU_1 \bV_1^{-1} \bU_1.$$ Since $\bH$ is positive definite by condition (ii), there exists a constant $c > 0$, such that with probability approaching one, $W\left(\widehat{\gamma}_{g}\right) \geq c\left\|\widehat{\gamma}_{g}\right\|^{2} / n$. Therefore, by applying Lemma \ref{l3}, we have $\xi_n \leq\left\|\widehat{\gamma}_{g}\right\|\left(1+\| \widehat{\gamma}_{g}\|\right) o_{p}(1)+\left\|\widehat{\gamma}_{g}\right\| O_{p}(1)-c\left\|\widehat{\gamma}_{g}\right\|^{2}+o_{p}(1)=\left[-c+o_{p}(1)\right]\|\widehat{\gamma}_{g}\|^{2}+\| \widehat{\gamma}_{g}\| O_{p}(1)+o_{p}(1)$.

As $\liminf_{n\to\infty}\xi_n \geq 0$ and $-c+o_{p}(1)<0$ with probability approaching one, it follows that $\left\|\widehat{\gamma}_{g}\right\|^{2}-2\left\|\widehat{\gamma}_{g}\right\| O_{p}(1) \leq o_{p}(1)$ for large enough $n$, with probability approaching 1. Hence for large enough $n$, $\left[\|\widehat{\gamma}_{g}\|-O_{p}(1)\right]^{2} \leq O_{p}(1)$ and $\big|\|\widehat{\gamma}_{g}\|-O_{p}(1) \big| \leq O_{p}(1)$, with probability approaching 1. This concludes $\|\widehat{\gamma}_{g}\|=O_{p}(1)$.
\end{proof}

\begin{proof}[Proof of Theorem \ref{t4}]
Let $Q_{n}(\gamma)=n\left[W_{n}(\gamma)-W_{n}(0)\right]$. Under the envelope setting, $Q_n(\gamma)$ is minimized at $\widehat{\gamma}_{g}$. Based on the results in Lemmas \ref{l3} and \ref{l4}, and Taylor expansion, we see that
\begin{align}
\begin{aligned} Q_{n}(\gamma) &=\sqrt{n} \widehat{\mathbf{D}}^{\T} \gamma+n W(\gamma)+o_{p}(1)=\sqrt{n} \widehat{\mathbf{D}}^{\T} \gamma+\frac{1}{2} \boldsymbol{\gamma}^{\T} \mathbf{H} \boldsymbol{\gamma}+o_{p}(1) \nonumber\\
& \stackrel{d}{\rightarrow} \mathbf{N}^{\T} \boldsymbol{\gamma}+\frac{1}{2} \boldsymbol{\gamma}^{\T} \mathbf{H} \gamma=: Q(\boldsymbol{\gamma}) \end{aligned},\nonumber
\end{align}
where $\mathbf{N}=N\left(0,4 \bU_1 \bV_1^{-1} \bU_1\right)$ because $\sqrt{n} \widehat{\mathbf{D}}=2 \sqrt{n} \bU_1 \widehat{\boldsymbol{\Delta}} q_{n}(0)$. Therefore, by Lemma \ref{l4} and the argmax theorem (Corollary 5.58 in \cite{V1998}), we have $\widehat{\gamma}_{g} \stackrel{d}{\rightarrow} \tilde{\gamma}$, where 
\begin{align}
\tilde{\gamma}=\argmin_{\frac{\gamma}{\sqrt{n}}+\theta^*=env\left(\zeta\right)}{Q(\boldsymbol{\gamma})}=\argmin_{\frac{\gamma}{\sqrt{n}}+\theta^*=env\left(\zeta\right)} \frac{1}{2n}\left(\gamma+\mathbf{H}^{-1} \mathbf{N}\right)^{\T} \mathbf{H}\left(\gamma+\mathbf{H}^{-1} \mathbf{N}\right).\nonumber
\end{align}
Since the parameter vector $\gamma$ is overparameterized, we next apply \cite{sh1986} to establish the asymptotic distribution of $\tilde{\gamma}$. We form the discrepancy function $F(x,\xi)$ in \cite{sh1986} as
\begin{align}\label{dis}
F(x, \xi)=\frac{1}{2}\left(\frac{\gamma}{\sqrt{n}}+\frac{\mathbf{H}^{-1} \mathbf{N}}{\sqrt{n}}\right)^{\T} \mathbf{H}\left(\frac{\gamma}{\sqrt{n}}+\frac{\mathbf{H}^{-1} \mathbf{N}}{\sqrt{n}}\right) 
\end{align}
where $x$ and $\xi$ in our context represent $-\mathbf{H}^{-1} \mathbf{N} / \sqrt{n}$  and $\gamma / \sqrt{n}$, respectively. It's easy to check that \eqref{dis} satisfies Shapiro's assumptions 1-5 and $\frac{\partial^{2} F}{\partial \xi \xi^{\T}}=\mathbf{H}=2 \bU_1 \bV_1^{-1} \bU_1$.  In addition, $-\mathbf{H}^{-1} \mathbf{N} \stackrel{d}{\rightarrow} N\left(0, \bU_1^{-1} \bV_1 \bU_1^{-1}\right)$. Let $\Psi_1 = \frac{\partial env(\zeta)}{\partial \zeta^{\T}}$. Therefore, by applying Proposition 4.1 of \cite{sh1986}, we have $\tilde{\gamma} \stackrel{d}{\rightarrow} N\left(0, \boldsymbol{\Lambda}_{g}\right)$, where
\begin{align}
\boldsymbol{\Lambda}_{g} = \Psi_{1}\left(\Psi_{1}^{\T} \mathbf{H} \Psi_{1}\right)^{\dagger} \Psi_{1}^{\T} \mathbf{H} \cdot \operatorname{avar}\left(-\mathbf{H}^{-1} \mathbf{N}\right) \cdot \mathbf{H} \Psi_{1}\left(\Psi_{1}^{\T} \mathbf{H} \Psi_{1}\right)^{\dagger} \Psi_{1}^{\T}=\Psi_{1}\left(\Psi_{1}^{\T} \bU_1 \bV_1^{-1} \bU_1 \Psi_{1}\right)^{\dagger} \Psi_{1}^{\T}.\nonumber
\end{align}
Hence $\widehat{\gamma}_{g}=\sqrt{n}\left(\widehat{\theta}_{g}-\theta^*\right) \stackrel{d}{\rightarrow} N\left(0, \boldsymbol{\Lambda}_{g}\right)$. Notice that $\bU_1 \bV_1^{-1} \bU_1 = \operatorname{avar}(\sqrt{n}\tilde\theta)^{-1}$, so the first part of Theorem \ref{t4} is proved.

Finally, to give a closed-form expression of $\Psi_1$, we introduce the contraction and expansion matrices that connect the $\ve$ and $\vech$ operators. For any symmetric matrix $M\in \real^{m\times m}$, let $C_m\in\real^{m(m+1)/2 \times m^2}$ and $E_m\in\real^{m^2 \times m(m+1)/2}$ be the unique contraction and expansion matrices such that $\vech(M) = C_m\ve(M)$ and $\ve(M) = E_m\vech(M)$ \citep{hs1979}. Then, it can be shown that $\Psi_1$ has the following expression $$ \begin{pmatrix}
  1 & 0 & 0 & 0 & 0 & 0\\
  0 &  \Gamma & \eta^{\T} \otimes \mathbf{I}_p & 0 & 0 & 0\\
  0 & 0 & 2C_p(\Gamma\Omega \otimes \mathbf{I}_p - \Gamma\otimes\Gamma_0\Omega_0\Gamma_0^{\T}) & C_p(\Gamma\otimes\Gamma)E_u & C_p(\Gamma_0 \otimes\Gamma_0)E_{p-u} & 0\\
  0 &  0 &  0 &  0 &  0 &  \mathbf{I}_p
 \end{pmatrix}.$$
For more detail on this result, see Appendix A.6 in \cite{cook2018}. Thus we have finished the proof of Theorem \ref{t4}.
\end{proof}

\begin{proof}[Proof of Corollary \ref{c1}]
To see the asymptotic efficiency of $\hat\theta_g$ relative to $\tilde\theta$, let $\Upsilon=\bU_1^{-1} \bV_1 \bU_1^{-1}$, we have
\begin{align}
\operatorname{avar}(\sqrt{n}\tilde\theta) - \operatorname{avar}(\sqrt{n}\hat\theta_g) &= \Upsilon-\Psi_1\left(\Psi_1^{\T} \Upsilon^{-1} \Psi_1\right)^{\dagger} \Psi_1^{\T}=\Upsilon^{1 / 2}\left(\mathbf{I}-\mathbf{P}_{\Upsilon^{-1 / 2} \Psi_1}\right) \Upsilon^{1 / 2} \nonumber\\
=\Upsilon^{1 / 2} \mathbf{Q}_{\Upsilon^{-1 / 2} \Psi_1}\Upsilon^{1 / 2} \succeq 0 \nonumber
\end{align}
where $\mathbf{P}_{\Upsilon^{-1 / 2} \Psi_1}=\Upsilon^{-1 / 2} \Psi_1\left(\Psi_1^{\T} \Upsilon^{-1} \Psi_1\right)^{\dagger} \Psi_1^{\T} \Upsilon^{-1 / 2}$ is the projection matrix onto the column span of $\Upsilon^{-1 / 2} \Psi_1$, and $\mathbf{Q}_{\Upsilon^{-1 / 2}\Psi_1}$ is the projection matrix onto the orthogonal complement of span of $\Upsilon^{-1 / 2} \Psi_1$. Since both of them are positive semidefinite, the proof is completed.
\end{proof}

\begin{proof}[Proof of Corollary \ref{c2}]
Let $\theta_1^* = (\mu^*,\beta^*,\vech(\Sigma_x^*))$ and let $\hat\theta_{g1}$ be the enveloped Huber estimator of $\theta_1^*$, i.e. $\hat\theta_{g1}$ contains the first $1+p+\frac{p(p+1)}{2}$ components of $\hat\theta_g$. Correspondingly, we can write $$\Psi_1 = \begin{pmatrix}
  \Psi & 0 \\
  0 &   \mathrm{I}_{p} 
 \end{pmatrix}, \bU_1 = \begin{pmatrix}
  \bU & 0 \\
  0 & \mathrm{I}_{p} 
 \end{pmatrix}, \text{ and\;} \bV_1 = \begin{pmatrix}
  \bV & \bA \\
  \bA^{\T} &  \bV_{1,33} \end{pmatrix}$$
where $\bA = (\mathbf{0}_{p\times(p+1)},\bV_{1,23}^{\T})^{\T}$, $\Psi$ and $\bU$ are the corresponding $\frac{p(p+1)}{2}+p+1$ dimensional matrices. Then, by the result of Theorem \ref{t4}, it can be directly verified that $$\sqrt{n}(\hat\theta_{g1} - \theta_1^*) \overset{d}\to N(\mathbf{0},\Psi\left(\Psi^{\T} \bU \bV^{-1} \bU \Psi \right)^{\dagger} \Psi^{\T}).$$

Recall we denote $\zeta = (\mu,\eta^{\T},\ve(\Gamma)^{\T},\vech(\Omega)^{\T},\vech(\Omega_0)^{\T},\mu_x^{\T})^{\T}$. And we have the relation $\theta = (\mu,\beta^{\T},\vech(\Sigma_x)^{\T},\mu_x^{\T})^{\T} = env(\zeta)$. We now let $\zeta_1 = (\mu,\eta^{\T},\ve(\Gamma)^{\T},\vech(\Omega)^{\T},\vech(\Omega_0)^{\T})^{\T}$ and $\theta_1 = (\mu,\beta^{\T},\vech(\Sigma_x)^{\T})^{\T} = env_1(\zeta_1) \coloneqq (\mu, (\Gamma\eta)^{\T}, \vech(\Gamma\Omega\Gamma^{\T}+\Gamma_0\Omega_0\Gamma_0^{\T})^{\T})^{\T}$. Then we have $\Psi = \frac{\partial env_1(\zeta_1)}{\partial \zeta_1^{\T}}$. By the expression of $\Psi_1$ given in Theorem \ref{t4}, we know that $\Psi$ has the following expression $$ \Psi = \begin{pmatrix}
  1 & 0 & 0 & 0 & 0 \\
  0 &  \Gamma & \eta^{\T} \otimes \mathbf{I}_p & 0 & 0\\
  0 & 0 & \Psi_{33} & C_p(\Gamma\otimes\Gamma)E_u & C_p(\Gamma_0 \otimes\Gamma_0)E_{p-u}
 \end{pmatrix}$$
where $\Psi_{33} = 2C_p(\Gamma\Omega \otimes \mathbf{I}_p - \Gamma\otimes\Gamma_0\Omega_0\Gamma_0^{\T})$.

When $\epsilon$ and $\bx$ are independent, we have $\ev_{\theta^*}[\psi(\epsilon)^2\bw\bw^{\T}] = \ev_{\theta^*}[\psi(\epsilon)^2]\ev_{\theta^*}[\bw\bw^{\T}]$, and also $\ev_{\theta^*}[\psi'(\epsilon)\bw\bw^{\T}] = \ev_{\theta^*}[\psi'(\epsilon)]\ev_{\theta^*}[\bw\bw^{\T}]$. Then $$\bU^{-1}\bV\bU^{-1} = \begin{pmatrix}
  \frac{\ev_{\theta^*}[\psi^2(\epsilon)]}{(\ev_{\theta^*}[\psi'(\epsilon)])^2}\ev_{\theta^*}[\bw\bw^{\T}]^{-1} & 0  \\
  0 &   \var_{\theta^*}\Big(\vech\big((\bx-\mu_x^*)(\bx-\mu_x^*)^{\T}\big)\Big)
 \end{pmatrix}.$$ By Theorem \ref{t2}, this is also the $\mathrm{avar}(\tilde\theta_1)$, where $\tilde\theta_1 = (\hat\mu_h,\hat\beta_h^{\T},\vech(S_x)^{\T})^{\T}$, i.e. the first $\frac{p(p+1)}{2}+p+1$ components of the Huber estimator $\tilde\theta$. In particular, this implies that $\sqrt{n}(\hat\beta_h - \beta^*)\overset{d}\to N(\mathbf{0}, \frac{\ev_{\theta^*}[\psi^2(\epsilon)]}{(\ev_{\theta^*}[\psi'(\epsilon)])^2} \Sigma_x^{*-1})$.

Let $\bJ = \bU\bV^{-1}\bU$. Notice that the asymptotic variance of $\hat\theta_{g1}$ depends on $\Psi$ only through its column space. Therefore we may replace $\Psi$ with any matrix $\Psi'$ which has the same column space as $\Psi$. Following \cite{cook2010}, we choose $\Psi' = \text{blockdiag}(1,\Psi_1')$ such that $\Psi'^{\T} \bJ \Psi'$ is block diagonal, where $$\Psi_1' =  \begin{pmatrix}
  \Gamma & \eta^{\T} \otimes \Gamma_0 & 0 & 0 \\
   0 & \Psi_{33}' & C_p(\Gamma\otimes\Gamma)E_u & C_p(\Gamma_0 \otimes\Gamma_0)E_{p-u}
 \end{pmatrix} \equiv \begin{pmatrix}
  \Psi_{1,1}' & \Psi_{1,2}' & \Psi_{1,3}' & \Psi_{1,4}' 
 \end{pmatrix} ,$$
where $\Psi_{33}' = 2C_p(\Gamma\Omega \otimes \Gamma_0 - \Gamma \otimes \Gamma_0\Omega_0)$. If we let $\Psi'' = \text{blockdiag}(1,\Psi_1'')$, where $$\Psi_1'' =  \begin{pmatrix}
  \mathbf{I}_u & \eta^{\T}\otimes \Gamma^{\T} & 0 & 0 \\
  0 &  \mathbf{I}_u \otimes \Gamma_0^{\T} & 0 & 0 \\
  0 &  2C_u (\Omega\otimes \Gamma^{\T})& \mathbf{I}_{u(u+1)/2} & 0\\
  0 &  0 & 0 & \mathbf{I}_{(p-u)(p-u+1)/2}
 \end{pmatrix},$$
then it can be seen that $\Psi'\Psi'' = \Psi$ and $\Psi''$ is a full rank square matrix. This confirms that $\Psi'$ has the same column space as $\Psi$. If $\ev_{\theta^*}[\bx] = 0$, we then have $$\bJ = \text{blockdiag}\left(\frac{(\ev_{\theta^*}[\psi'(\epsilon)])^2}{\ev_{\theta^*}[\psi^2(\epsilon)]}, \frac{(\ev_{\theta^*}[\psi'(\epsilon)])^2}{\ev_{\theta^*}[\psi^2(\epsilon)]} \Sigma_x^*, \var_{\theta^*}\Big(\vech\big(\bx\bx^{\T}\big)\Big)^{-1}\right).$$ Write $\bJ = \text{blockdiag}(\frac{(\ev_{\theta^*}[\psi'(\epsilon)])^2}{\ev_{\theta^*}[\psi^2(\epsilon)]},\bJ_1)$. When $\bx$ is normal, the moment estimator of $\Sigma_x^*$ in this case is asymptotically equivalent to the MLE of $\Sigma_x^*$ with the Fisher information matrix $\frac{1}{2}E_p^{\T}(\Sigma_x^{*-1}\otimes \Sigma_x^{*-1})E_p$ \citep{cook2013}. Thus $\bJ_1 =  \text{blockdiag}\left(\frac{(\ev_{\theta^*}[\psi'(\epsilon)])^2}{\ev_{\theta^*}[\psi^2(\epsilon)]} \Sigma_x^*, \frac{1}{2}E_p^{\T}(\Sigma_x^{*-1}\otimes \Sigma_x^{*-1})E_p\right)$. After matrix multiplication, one can verify that $\Psi_1'^{\T}\bJ_1\Psi_1'$ is block diagonal matrix, and so is $\Psi'^{\T}\bJ\Psi'$. Therefore, we have $\mathrm{avar}(\hat\beta_{eh}) = \Gamma(\Psi_{1,1}'^{\T}\bJ_1\Psi_{1,1}')^{\dagger}\Gamma^{\T} + (\eta^{\T} \otimes \Gamma_0) (\Psi_{1,2}'^{\T}\bJ_1\Psi_{1,2}')^{\dagger}(\eta \otimes \Gamma_0^{\T})$. By the expression of $\Psi'$ and $\bJ_1$, we have $\Psi_{1,1}'^{\T}\bJ_1\Psi_{1,1}' = \Gamma^{\T}\frac{(\ev_{\theta^*}[\psi'(\epsilon)])^2}{\ev_{\theta^*}[\psi^2(\epsilon)]}\Sigma_x^* \Gamma = \frac{(\ev_{\theta^*}[\psi'(\epsilon)])^2}{\ev_{\theta^*}[\psi^2(\epsilon)]}\Omega$. For the term $\Psi_{1,2}'^{\T}\bJ_1\Psi_{1,2}'$, we have $\Psi_{1,2}'^{\T}\bJ_1\Psi_{1,2}' = \bZ_1 + \bZ_2$, where $\bZ_1 = (\eta \otimes \Gamma_0^{\T}) \frac{(\ev_{\theta^*}[\psi'(\epsilon)])^2}{\ev_{\theta^*}[\psi^2(\epsilon)]}\Sigma_x^* (\eta^{\T} \otimes \Gamma_0)$ and $\bZ_2= 2(\Gamma\Omega\otimes\Gamma_0 - \Gamma\otimes\Gamma_0\Omega_0)^{\T}C_p^{\T}E_p^{\T}(\Sigma_x^{*-1}\otimes \Sigma_x^{*-1})E_p C_p(\Gamma\Omega\otimes\Gamma_0 - \Gamma\otimes\Gamma_0\Omega_0)$. Notice that we have $\Sigma_x^* = \Gamma\Omega\Gamma^{\T} + \Gamma_0\Omega_0\Gamma_0^{\T}$, so $\Gamma^{\T}\Sigma_x^*\Gamma_0 = 0$ and $\Sigma_x^{*-1} = \Gamma\Omega^{-1}\Gamma^{\T} + \Gamma_0\Omega_0^{-1}\Gamma_0^{\T}$. Then we have
\begin{align}
\bZ_1 &= (\eta \otimes \Gamma_0^{\T}) \frac{(\ev_{\theta^*}[\psi'(\epsilon)])^2}{\ev_{\theta^*}[\psi^2(\epsilon)]}\Sigma_x^* (\eta^{\T} \otimes \Gamma_0) \n\\
&=\frac{(\ev_{\theta^*}[\psi'(\epsilon)])^2}{\ev_{\theta^*}[\psi^2(\epsilon)]} (\eta\eta^{\T} \otimes \Gamma_0^{\T}\Sigma_x^*\Gamma_0)\n\\
&=\frac{(\ev_{\theta^*}[\psi'(\epsilon)])^2}{\ev_{\theta^*}[\psi^2(\epsilon)]} (\eta\eta^{\T} \otimes \Omega_0),\n
\end{align}
and by Appendix A.6 in \cite{cook2018},
\begin{align}
\bZ_2 &= 2(\Gamma\Omega\otimes\Gamma_0 - \Gamma\otimes\Gamma_0\Omega_0)^{\T}C_p^{\T}E_p^{\T}(\Sigma_x^{*-1}\otimes \Sigma_x^{*-1})E_p C_p(\Gamma\Omega\otimes\Gamma_0 - \Gamma\otimes\Gamma_0\Omega_0) \n\\
&=2(\Omega\Gamma^{\T}\otimes\Gamma_0^{\T} - \Gamma^{\T}\otimes\Omega_0\Gamma_0^{\T})(\Sigma_x^{*-1}\otimes \Sigma_x^{*-1})E_p C_p(\Gamma\Omega\otimes\Gamma_0 - \Gamma\otimes\Gamma_0\Omega_0)\n\\
&=2(\Omega\Gamma^{\T}\otimes\Gamma_0^{\T} - \Gamma^{\T}\otimes\Omega_0\Gamma_0^{\T})(\Sigma_x^{*-1}\otimes \Sigma_x^{*-1})\bP_{E_p}(\Gamma\Omega\otimes\Gamma_0 - \Gamma\otimes\Gamma_0\Omega_0)\n\\
&=2(\Gamma^{\T}\otimes\Omega_0^{-1}\Gamma_0^{\T} - \Omega^{-1}\Gamma^{\T}\otimes\Gamma_0^{\T})\bP_{E_p}(\Gamma\Omega\otimes\Gamma_0 - \Gamma\otimes\Gamma_0\Omega_0)\n\\
& = \Gamma^{\T}\Gamma\Omega\otimes\Omega_0^{-1}\Gamma_0^{\T}\Gamma_0 - \Omega^{-1}\Gamma^{\T}\Gamma\Omega\otimes\Gamma_0^{\T}\Gamma_0 - \Gamma^{\T}\Gamma\otimes\Omega_0^{-1}\Gamma_0^{\T}\Gamma_0\Omega_0 + \Omega^{-1}\Gamma^{\T}\Gamma\otimes\Gamma_0^{\T}\Gamma_0\Omega_0\n\\
&= \Omega\otimes\Omega_0^{-1} + \Omega^{-1}\otimes\Omega_0 - 2\mathbf{I}_u \otimes\mathbf{I}_{p-u}.\n
\end{align}
Therefore, $\Psi_{1,2}'^{\T}\bJ_1\Psi_{1,2}' = \frac{(\ev_{\theta^*}[\psi'(\epsilon)])^2}{\ev_{\theta^*}[\psi^2(\epsilon)]} (\eta\eta^{\T} \otimes \Omega_0) +  \Omega\otimes\Omega_0^{-1} + \Omega^{-1}\otimes\Omega_0 - 2\mathbf{I}_u \otimes\mathbf{I}_{p-u}$, thus \begin{align}
&\mathrm{avar}(\hat\beta_{eh}) =\frac{\ev_{\theta^*}[\psi^2(\epsilon)]}{(\ev_{\theta^*}[\psi'(\epsilon)])^2} \Gamma\Omega^{-1}\Gamma^{\T} \n\\
&+ (\eta^{\T} \otimes \Gamma_0) \left(\frac{(\ev_{\theta^*}[\psi'(\epsilon)])^2}{\ev_{\theta^*}[\psi^2(\epsilon)]} (\eta\eta^{\T} \otimes \Omega_0)+ \Omega\otimes\Omega_0^{-1} + \Omega^{-1}\otimes\Omega_0 - 2\mathbf{I}_u \otimes\mathbf{I}_{p-u}\right)^{\dagger}(\eta \otimes \Gamma_0^{\T}).\n
\end{align}
So we finish the proof of Corollary \ref{c2}.

\end{proof}

\section*{Acknowledgements}
Zou is supported in part by NSF grants DMS-1915842 and DMS-2015120.

\end{document}